\documentclass[a4paper,twocolumn,11pt,accepted=2024-02-13]{quantumarticle}
\pdfoutput=1
\usepackage[utf8]{inputenc}
\usepackage[english]{babel}
\usepackage[T1]{fontenc}
\usepackage{amsmath}
\usepackage{hyperref}

\usepackage{tikz}
\usepackage{lipsum}

\usepackage{graphicx}
\usepackage{dcolumn}
\usepackage{amssymb,amsfonts,amsthm}
\usepackage{mathtools} 
\usepackage{txfonts}
\usepackage{bbm} 
\usepackage{bm} 
\usepackage{mathrsfs} 
\usepackage{xcolor}
\usepackage[numbers,sort&compress]{natbib}
\usepackage{braket}
\DeclareMathOperator{\Tr}{Tr}
\allowdisplaybreaks
\usepackage{multicol,blindtext}

\newtheorem{theorem}{Theorem}

\begin{document}

\title{Entanglement catalysis for quantum states and noisy channels}

\author{Chandan Datta}
\orcid{0000-0002-0069-4597}
\email{dattachandan10@gmail.com}
\address{Centre for Quantum Optical Technologies, Centre of New Technologies,
University of Warsaw, Banacha 2c, 02-097 Warsaw, Poland}
\address{Institute for Theoretical Physics III, Heinrich Heine University D\"{u}sseldorf, Universit\"{a}tsstra{\ss}e 1, D-40225 D\"{u}sseldorf, Germany}
\address{Department of Physics, Indian Institute of Technology Jodhpur, Jodhpur 342030, India}

\author{Tulja Varun Kondra}
\address{Centre for Quantum Optical Technologies, Centre of New Technologies,
University of Warsaw, Banacha 2c, 02-097 Warsaw, Poland}

\author{Marek Miller}
\address{Centre for Quantum Optical Technologies, Centre of New Technologies,
University of Warsaw, Banacha 2c, 02-097 Warsaw, Poland}

\author{Alexander Streltsov}
\orcid{0000-0002-7742-5731}
\address{Centre for Quantum Optical Technologies, Centre of New Technologies,
University of Warsaw, Banacha 2c, 02-097 Warsaw, Poland}

\begin{abstract}
    Many applications of the emerging quantum technologies, such as  quantum teleportation and quantum key distribution, require singlets, maximally entangled states of two quantum bits. It is thus of utmost importance to develop optimal procedures for establishing singlets between remote parties. As has been shown very recently, singlets can be obtained from other quantum states by using a quantum catalyst, an entangled quantum system which is not changed in the procedure. In this work we take this idea further, investigating properties of entanglement catalysis and its role for quantum communication. For transformations between bipartite pure states, we prove the existence of a universal catalyst, which can enable all possible transformations in this setup. We demonstrate the advantage of catalysis in asymptotic settings, going beyond the typical assumption of independent and identically distributed systems. We further develop methods to estimate the number of singlets which can be established via a noisy quantum channel when assisted by entangled catalysts. For various types of quantum channels our results lead to optimal protocols, allowing to establish the maximal number of singlets with a single use of the channel. 
\end{abstract}

\maketitle

\section{Introduction}\label{intro}

Quantum catalysis enhances the abilities of remote parties to manipulate entangled systems via local operations and classical communication (LOCC)~\cite{JonathanPhysRevLett.83.3566,EisertPhysRevLett.85.437,Kondra2102.11136,Lipka-Bartosik2102.11846}. Two remote parties, Alice and Bob, can convert a shared quantum state $\ket{\psi}^{AB}$ into another state $\ket{\phi}^{AB}$ whenever the reduced states $\psi^A$ and $\phi^A$ fulfill the majorization relation $\psi^{A}\prec\phi^{A}$~\cite{NielsenPhysRevLett.83.436}. 
If this condition is violated, Alice and Bob cannot convert $\ket{\psi}^{AB}$ into $\ket{\phi}^{AB}$ via LOCC. However, in some cases Alice and Bob can still convert $\ket{\psi}^{AB}$ into $\ket{\phi}^{AB}$ by using catalysis. A quantum catalyst is an additional quantum system in an entangled state $\ket{\tau}^{A'B'}$, enabling the LOCC conversion 
\begin{equation}
    \ket{\psi}^{AB}\otimes \ket{\tau}^{A'B'}\rightarrow \ket{\phi}^{AB}\otimes \ket{\tau}^{A'B'}. \label{eq:Catalysis}
\end{equation}
After the first example of entanglement catalysis was presented~\cite{JonathanPhysRevLett.83.3566}, this topic has been explored in more detail over the last decades~\cite{VidalPhysRevA.62.012304, PhysRevA.64.042314, DuanPhysRevA.71.042319, Turgut_2007, Klimesh0709.3680, Aubrun2008, SandersPhysRevA.79.054302, GraboweckyPhysRevA.99.052348,Kondra2102.11136,Lipka-Bartosik2102.11846,Gupta_PRA2022}. We refer to \cite{catalysis_review} for a general overview of catalysis. 

Despite significant efforts~\cite{VidalPhysRevA.62.012304, PhysRevA.64.042314, DuanPhysRevA.71.042319, Turgut_2007, Klimesh0709.3680, Aubrun2008, SandersPhysRevA.79.054302, GraboweckyPhysRevA.99.052348}, no simple criteria are known for exact catalytic transformations between two given pure bipartite entangled states, while a significant progress has been achieved for approximate weakly correlated entanglement catalysis~\cite{Kondra2102.11136,Lipka-Bartosik2102.11846}. An approximate weakly correlated entanglement catalysis from $\ket{\psi}^{AB}$ to $\ket{\phi}^{AB}$ allows for an error in the final state, assuming that the error can be made arbitrarily small by choosing an appropriate catalyst state. As was shown in~\cite{Kondra2102.11136}, approximate weakly correlated entanglement catalysis between bipartite pure states is fully characterized by entanglement entropy of the initial and the final states. In particular, Alice and Bob can transform $\ket{\psi}^{AB}$ into $\ket{\phi}^{AB}$ iff~\cite{Kondra2102.11136}
\begin{equation}
    S(\psi^A) \geq S(\phi^A), \label{eq:CatalysisPure}
\end{equation}
with the von Neumann entropy $S(\rho) = -\Tr[\rho \log_2 \rho]$.

For any two states fulfilling Eq.~(\ref{eq:CatalysisPure}), it is possible to find a catalyst by using the methods presented in~\cite{Kondra2102.11136,Lipka-Bartosik2102.11846}. The  constructed catalyst depends on the particular states $\ket{\psi}^{AB}$ and $\ket{\phi}^{AB}$, which makes it useful only for this particular transition. It is an open problem whether \emph{universal catalysis} exists, i.e., whether there are quantum states that can catalyze transitions for all quantum states fulfilling Eq.~(\ref{eq:CatalysisPure}). Here, we give a positive answer to this question. Furthermore, we explore the role of catalysis in asymptotic conversion. The advantage of catalysis for setups beyond the usual independent and identically distributed (iid) scenario is also discussed. 

Quantum information science strives to understand the limitations of sending quantum systems over long distances. As any practical quantum channel is necessarily subject to noise, knowing how many qubits can be reliably transmitted via a noisy channel is of crucial importance for the development of quantum technologies. The standard approach to this problem assumes that the communicating parties have unrestricted access to the channel. Then, the problem reduces to the evaluation of quantum capacity, a quantity which captures the number of reliably communicated qubits per channel use~\cite{LloydPhysRevA.55.1613,Divincenzo_1998,Barnum_1998,Schumacher_1998,Devetak1377491}.  
In this work, we introduce \emph{catalytic communication} where the communicating parties make use of a noisy quantum channel only once, and have access to LOCC and entangled catalysts. We develop methods of estimating the \emph{catalytic capacity} of the channel, corresponding to the number of qubits which can be sent reliably in the presence of a catalyst. 
For general quantum channels we provide upper and lower bounds on the catalytic capacity of the channel.

\section{Catalytic transformations of entangled states}

The original definition of entanglement catalysis in Eq.~(\ref{eq:Catalysis}) assumes that the transformation is exact, i.e., the final state $\ket{\phi}^{AB}$ is obtained perfectly without any error. Recently, this definition has been extended to approximate transformations~\cite{Kondra2102.11136,Lipka-Bartosik2102.11846}. In this more general case, an error is allowed in the final state, provided that the error can be made arbitrarily small. We say that an \emph{approximate weakly correlated catalytic transformation} $\ket{\psi}^{AB}\rightarrow\ket{\phi}^{AB}$ is possible if and only if for any $\varepsilon > 0$ there exists a (not necessarily pure) catalyst state $\tau^{A'B'}$ and an LOCC protocol $\Lambda$ such that~\cite{Kondra2102.11136}
\begin{align}
\left\Vert \Lambda\left(\psi^{AB}\otimes\tau^{A'B'}\right)-\phi^{AB}\otimes\tau^{A'B'}\right\Vert _{1}  \leq\varepsilon,\label{eq:CatalyticLOCC2}\\
\mbox{Tr}_{AB}\left[\Lambda\left(\psi^{AB}\otimes\tau^{A'B'}\right)\right]  =\tau^{A'B'}.\label{eq:CatalyticLOCC1}
\end{align}
Here, we use the notation $\psi^{AB} = \ket{\psi}\!\bra{\psi}^{AB}$ and similar for $\phi^{AB}$. Note that the final state $\sigma=\Lambda (\psi^{AB}\otimes\tau^{A'B'})$ can, in principle, exhibit correlations between the system ($AB$) and the catalyst ($A'B'$). However, Eq.~(\ref{eq:CatalyticLOCC2}) ensures that the amount of these correlations can be made arbitrarily small, i.e., the system effectively decouples from the catalyst. The decoupling condition in Eq.~(\ref{eq:CatalyticLOCC2}) also implies a bound on the mutual information between the primary system $AB$ and the catalyst $A'B'$~\cite{Rubboli2111.13356}. Eq.~(\ref{eq:CatalyticLOCC1}) means that the catalyst is unchanged in the process, avoiding any undesired effects due to ``embezzling'' of entanglement~\cite{vanDamPhysRevA.67.060302}.

Alternatively, we can say that there is an approximate catalytic transformation from $\ket{\psi}^{AB}$ to $\ket{\phi}^{AB}$ if there exists a sequence of catalyst states $\{\tau_n^{A'B'}\}$ and a sequence of LOCC protocols $\{\Lambda_n\}$ such that~\cite{Kondra2102.11136}
\begin{align}
\lim_{n\rightarrow\infty}\left\Vert \Lambda_{n}\left(\psi^{AB}\otimes\tau_{n}^{A'B'}\right)-\phi^{AB}\otimes\tau_{n}^{A'B'}\right\Vert _{1}& =0,\label{eq:CatalyticLOCC4}\\
\quad\quad\mbox{Tr}_{AB}\left[\Lambda_{n}\left(\psi^{AB}\otimes\tau_{n}^{A'B'}\right)\right]  =\tau_{n}^{A'B'}.\label{eq:CatalyticLOCC3}
\end{align}
It is immediately clear that these conditions are equivalent to Eqs.~(\ref{eq:CatalyticLOCC2}) and (\ref{eq:CatalyticLOCC1}). As has been shown very recently in~\cite{Kondra2102.11136}, an approximate catalytic transformation from $\ket{\psi}^{AB}$ into $\ket{\phi}^{AB}$ is possible if and only if Eq.~(\ref{eq:CatalysisPure}) is fulfilled. In the rest of this article, whenever we refer to ``catalytic transformations'' we mean approximate weakly correlated catalytic transformations as defined in Eqs.~(\ref{eq:CatalyticLOCC2}) and (\ref{eq:CatalyticLOCC1}).

\subsection{Dimension of catalyst system}
Whenever we talk about state transformations using a catalyst, the question of its dimension comes into the picture. In the original definition of entanglement catalysis~\cite{JonathanPhysRevLett.83.3566} a catalyst of finite dimension has been used. In contrast, we will now see that for arbitrary precision, approximate weakly correlated catalysis requires a catalyst of unbounded dimension in general. To show this, we first introduce logarithmic negativity~\cite{ZyczkowskiPhysRevA.58.883,Vidal_2002_PRA}
\begin{equation}
    E_N(\rho)=\log_2\Vert \rho^{T_A}\Vert_1,
\end{equation}
where $\rho = \rho^{AB}$ is a bipartite state, and $T_A$ represents partial transpose with respect to the subsystem $A$. For pure states, we will also consider the entanglement entropy defined as $E(\ket{\psi}) = S(\psi^A)$~\cite{BennettPhysRevA.53.2046,VedralPhysRevLett.78.2275,HorodeckiRevModPhys.81.865}. Recall that logarithmic negativity is additive on tensor products~\cite{Vidal_2002_PRA}
\begin{equation}
    E_N(\rho\otimes\sigma)=E_N(\rho)+E_N(\sigma).
\end{equation}
Consider now two bipartite pure states $\ket{\psi}$ and $\ket{\phi}$ such that 
\begin{subequations} \label{eq:UnboundedDimension3}
\begin{align}
E(\ket{\psi}) & \geq E(\ket{\phi}), \label{eq:UnboundedDimension1}\\
E_{N}(\ket{\psi}) & <E_{N}(\ket{\phi}). \label{eq:UnboundedDimension2}
\end{align}
\end{subequations}
Examples for such states are provided in Appendix \ref{sec:different_ordering}. Eq.~(\ref{eq:UnboundedDimension1}) implies that $\ket{\psi}$ can be converted into $\ket{\phi}$ by catalytic LOCC~\cite{Kondra2102.11136}. However, as we will see in the following, Eq.~(\ref{eq:UnboundedDimension2}) means that the catalyst must have an unbounded dimension, in order for the conversion to work with arbitrary precision. To see this, consider the transformation  $\psi\otimes\tau_n\rightarrow\phi\otimes\tau_n$, with the properties given in Eqs.~(\ref{eq:UnboundedDimension3}). Remembering that logarithmic negativity is monotonic under LOCC, we have 
\begin{equation}\label{monotonicity_negativity}
E_N(\psi\otimes\tau_n)\geq E_N(\Lambda[\psi\otimes\tau_n])
\end{equation}
for any LOCC protocol $\Lambda$. Furthermore, $\Vert \mu_n -\phi\otimes\tau_n\Vert_1$ is arbitrarily close to zero for some large $n$, where $\mu_n=\Lambda(\psi\otimes\tau_n)$. Now, if the dimension of the catalyst is finite then the logarithmic negativity of $\mu_n$ is arbitrarily close to $\phi\otimes\tau_n$, see Appendix~\ref{sec:ContinuityNegativity} for the proof. Therefore, we have
\begin{equation}
E_N(\psi\otimes\tau_n)\geq E_N(\phi\otimes\tau_n) \Rightarrow E_N(\ket{\psi}) \geq E_N(\ket{\phi}), 
\end{equation}
where in the second step we use the additivity. Hence, we arrive at a contradiction, and as a result, we need an unbounded catalyst to achieve the above transformation. Note that here we consider that the state transformation occurs with arbitrary precision. Otherwise, we can utilise a finite-dimensional catalyst, the dimension of which is determined by the target precision. This kind of phenomenon has also been studied in \cite{Rubboli2111.13356} for general resource theories.

We will now go one step further and prove that for a catalytic transformation from $\ket{\psi}$ to $\ket{\phi}$ correlations between the primary system and catalysis are required if the states fulfill Eqs.~(\ref{eq:UnboundedDimension3}). For this, consider Eq. (\ref{monotonicity_negativity}) and suppose that for all $n$ greater than some $m$, there is no correlation between the system and the catalyst, i.e., $\Lambda(\psi\otimes\tau_n)=\phi'\otimes\tau_n$, where $\phi'$ is arbitrarily close to the desired state $\phi$ in trace distance. Using additivity and continuity of logarithmic negativity, we reach the following contradiction: $E_N(\ket{\psi}) \geq E_N(\ket{\phi})$. Therefore, a catalytic transformation for states fulfilling Eqs.~(\ref{eq:UnboundedDimension3}) requires correlations between the system and the catalyst. As discussed above, these correlations can be made arbitrarily small. 

\subsection{Correlations established by reusing the catalyst}
The correlations between the system and the catalyst also establish correlations across different systems, if the same catalyst is used repeatedly. We will investigate these correlations in the following. To simplify the notation, we will denote the primary system by $S$, and the catalyst will be denoted by $C$. Consider a system in a state $\rho^{S_1}$ is transformed arbitrarily close to $\sigma^{S_1}$ using a catalyst $\tau^C$ and an LOCC operation $\Lambda$ on the system and the catalyst, such that
\begin{eqnarray}
||\mu^{S_1C}-\sigma^{S_1}\otimes\tau^C||_1<\varepsilon \label{eq:CatalysisCorrelations}
\end{eqnarray}
with $\mu^{S_1C} = \Lambda[\rho^{S_1} \otimes \tau^C]$. Now, in the next step we use the same catalyst to convert the same state $\rho^{S_2}$ of another system $S_2$ into $\sigma^{S_2}$. The final state of the systems and the catalyst can be expressed as
\begin{equation}
    \mu^{S_1CS_2}=(\openone^{S_1}\otimes\Lambda^{CS_2})\left(\mu^{S_1C}\otimes\rho^{S_2}\right),
\end{equation}
where we demand that $\mbox{Tr}_{S_1S_2}[\mu^{S_1CS_2}]=\tau^C$ and 
\begin{equation}
\label{eq:use_catalysis}
\left\Vert \mbox{Tr}_{S_1}\left[\mu^{S_1CS_2}\right]-\sigma^{S_2}\otimes\tau^{C}\right\Vert _{1} < \varepsilon.
\end{equation}
Since, the system $S_1$ is correlated with the catalyst, the systems $S_1$ and $S_2$ will also be correlated. We are interested in obtaining an upper bound on this correlation, to be precise, we want an upper bound on $||\mbox{Tr}_C[\mu^{S_1CS_2}]-\sigma^{S_1}\otimes\sigma^{S_2}||_1$. Using the properties of the trace norm, Eq.~(\ref{eq:CatalysisCorrelations}) directly implies that 
\begin{equation}
\left\Vert \mu^{S_1C}\otimes\rho^{S_2}-\sigma^{S_1}\otimes\tau^{C}\otimes\rho^{S_2}\right\Vert _{1}<\varepsilon.
\end{equation}
Applying the LOCC procedure $\Lambda^{CS_2}$ onto the catalyst $C$ and the system $S_2$ and using the fact that the trace norm does not increase under quantum operations we further obtain
\begin{align}
& \left\Vert \Lambda^{CS_2}\left[\mu^{S_1C}\otimes\rho^{S_2}\right]-\Lambda^{CS_2}\left[\sigma^{S_1}\otimes\tau^{C}\otimes\rho^{S_2}\right]\right\Vert _{1}\nonumber \\ &  < \varepsilon,
\end{align}
which is equivalent to 
\begin{equation}
\left\Vert \mu^{S_1S_2C}-\sigma^{S_1}\otimes\mu^{S_2C}\right\Vert _{1}<\varepsilon.
\end{equation}
Now using triangle inequality, we find
\begin{align}
\label{eq:catalyst_two_before_partial}
&\left\Vert \mu^{S_1S_2C}-\sigma^{S_1}\otimes\sigma^{S_2}\otimes\tau^{C}\right\Vert _{1}  \nonumber\\
&\leq\left\Vert \mu^{S_1S_2C}-\sigma^{S_1}\otimes\mu^{S_2C}\right\Vert _{1}+\nonumber\\
&\quad\,\left\Vert \sigma^{S_1}\otimes\mu^{S_2C}-\sigma^{S_1}\otimes\sigma^{S_2}\otimes\tau^{C}\right\Vert _{1}\nonumber\\ 
  &< 2\varepsilon.
\end{align}
As trace norm does not increase under partial trace, we note that
\begin{equation}\label{eq:catalyst_two}
\left\Vert \mbox{Tr}_{C}\left[\mu^{S_1S_2C}\right]-\sigma^{S_1}\otimes\sigma^{S_2}\right\Vert _{1}<2\varepsilon.
\end{equation}

By repeating the same argument, this time taking Eq.~\eqref{eq:catalyst_two_before_partial} instead of Eq.~\eqref{eq:CatalysisCorrelations} as the starting point,
and utilising the same catalyst state $\tau^{C}$ as in Eq.~\eqref{eq:use_catalysis},
we can see that
\begin{equation}\label{three_catalyst}
   \left\Vert \mu^{S_1 S_2 S_3 C}-\sigma^{S_1}\otimes\sigma^{S_2}\otimes\sigma^{S_3}\otimes\tau^{C}\right\Vert _{1} < 3 \varepsilon.
\end{equation}
We can extend this reasoning to $n$ successive applications of the catalyst $\tau^{C}$ to obtain

\begin{equation}\label{n_catalyst}
   \left\Vert \mu^{S_1\cdots S_n C}-\sigma^{S_1}\otimes\cdots\otimes\sigma^{S_n}\otimes\tau^{C}\right\Vert _{1} < n \varepsilon,
\end{equation}
which implies that
\begin{equation}
    \left\Vert \mbox{Tr}_{C}\left[\mu^{S_1\cdots S_n C}\right]-\sigma^{S_1}\otimes\cdots\otimes\sigma^{S_n}\right\Vert _{1}<n\varepsilon.
\end{equation}

As we can see, the amount of correlation between systems $S_1, \ldots, S_n$ 
remains small after utilising the catalyst $\tau^{C}$ a finite number of times.

\section{Universal catalysis for bipartite pure states}
While the results in~\cite{Kondra2102.11136} prove that a catalytic conversion between bipartite states $\ket{\psi}^{AB}$ and $\ket{\phi}^{AB}$ is possible in principle if $S(\psi^A) \geq S(\phi^A)$, the state of the catalyst presented in~\cite{Kondra2102.11136} is not universal, as it strongly depends on the initial and the final state. It has remained an open problem in~\cite{Kondra2102.11136} whether a universal catalyst state exists, which can enable all catalytic transformations at once. We will close this gap with the following theorem, showing that universal catalysis is indeed possible for transitions between bipartite pure states.

\begin{theorem} \label{thm:UniversalCatalysis}
Consider a bipartite Hilbert space of arbitrary but finite dimension. For every $\varepsilon>0$
there exists a universal catalyst state $\tau_{\varepsilon}$ such that for every pair of pure states $\ket{\psi}^{AB}$ and $\ket{\phi}^{AB}$ with $S(\psi^A) \geq S(\phi^A)$ there is an LOCC protocol $\Lambda$ for which 
\begin{subequations} \label{eq:UniversalCatalysis}
\begin{align}
\left\Vert \Lambda\left(\psi^{AB}\otimes\tau_{\varepsilon}\right)-\phi^{AB}\otimes\tau_{\varepsilon}\right\Vert _{1}  <\varepsilon, \\
\mathrm{Tr}_{AB}\left[\Lambda\left(\psi^{AB}\otimes\tau_{\varepsilon}\right)\right]  =\tau_{\varepsilon}.
\end{align}
\end{subequations}
\end{theorem}
\begin{proof}
It is enough to prove the theorem for the case that $\ket{\psi}$ and $\ket{\phi}$ have the same Schmidt basis. To see this, note that Alice and Bob can always apply local unitaries to the initial state $\ket{\psi}$, ensuring that the Schmidt basis of $\ket{\psi}$ coincides with the Schmidt basis of $\ket{\phi}$.

First, we will show that there is a \textit{finite} collection of states $\{ \ket{\phi_{s}} \}$ such that  for every state  $\ket{\phi}$ in the considered Hilbert space $\mathcal {H}_{AB}$ there exists $\ket{\phi_{i}}$ for which 
\begin{subequations} \label{eq:mm_UniversalGridAll}
\begin{align}
\left\Vert \, \ket{\phi_{i}}\!\bra{\phi_{i}}-\ket{\phi}\!\bra{\phi} \,\right\Vert _{1} & < \frac{\varepsilon}{10}, \label{eq:mm_UniversalGrid}\\
\phi^A & \prec \phi_{i}^A, \label{eq:mm_UniversalMajorization}
\end{align}
\end{subequations}
where we consider strict majorization in Eq.~\eqref{eq:mm_UniversalMajorization}. To show this, let $U_{\phi}$ be a collection of open subsets of the set of pure states, defined in the following way.
If $\ket{\phi} = \sum_i \sqrt{c_i} \ket{ii}$, by $V_{\phi}$ we denote the set of all states that $\phi^A$ majorizes \textit{strictly}, i.e.
$\sum_i \sqrt{d_i} \ket{i}\ket{i} \in V_{\phi}$ iff
$\sum_{i=1}^{k} c_i > \sum_{i=1}^{k} d_i$,
for every $k=1,2,\ldots, d-1$, where $d = \min\{\dim \mathcal{H}_A, \dim \mathcal{H}_B\}$. Note that the coefficients $\{c_i\}$ and $\{d_i\}$ are sorted in non-increasing order.
Let also $B_{\phi} = \left\{ \ket{\psi}: || \ket{\psi}\! \bra{\psi} - \ket{\phi}\!\bra{\phi} ||_1 < \frac{\varepsilon}{10} \right\}$.
If $\ket{\phi}$ is a product state, let $U_{\phi} = B_{\phi}$;
otherwise $U_{\phi} = B_{\phi} \cap V_{\phi}$.
Clearly, each $U_{\phi}$ is open, and moreover the whole state space is covered by the collection (Note that for the collection to cover the whole set, it is necessary to single out the special case of $\ket{\phi}$ being a product state). Since the set of states is compact, there is a finite number of states $\{\ket{\phi_s}\}$ such that the sets $U_{\phi_s}$ cover the entire set of states.

Notice that by increasing slightly the largest Schmidt coefficient of a state $\ket{\phi_s}$, we preserve the majorization condition in Eq.~\eqref{eq:mm_UniversalMajorization}.
This way,
by imposing an even smaller threshold in Eq. \eqref{eq:mm_UniversalGrid},
and because the number of states $\ket{\phi_s}$ is finite,
we can always choose them to be such that 
\begin{align}
    &L = \nonumber\\ &\min \Big\{ |S(\phi_{i}^{A}) - S(\phi_{j}^{A})|,
    \, i,j: S(\phi_{i}^{A}), S(\phi_{j}^{A})> 0\Big\} \nonumber\\
    &> 0,
\end{align}
i.e., the entropies of the reduced states $\phi_{s}^A$ are all different as long as they are strictly positive.

Knowing the number $L$, we can repeat the above argument and show that there is also a finite collection $\ket{\psi_r}$ of states such that for every state  $\ket{\psi} \in \mathcal {H}_{AB}$ there exists $\ket{\psi_{j}}$ for which 
\begin{subequations} \label{eq:mm_UniversalGridAllInitial}
\begin{align}
\left\Vert \, \ket{\psi_{j}}\!\bra{\psi_{j}}-\ket{\psi}\!\bra{\psi} \,\right\Vert _{1} & < \delta, \label{eq:mm_UniversalGridInitial}\\
\psi^A & \prec \psi_{j}^A, \label{eq:mm_UniversalMajorizationInitial}
\end{align}
\end{subequations}
where $\delta>0$ is a number such that $\delta < \frac{\varepsilon}{4}$ and
\begin{equation}
\frac{\delta}{2} \log(d-1)+h\left(\frac{\delta}{2}\right) < L.
\end{equation}
The last inequality, together with the Fannes-Audenaert inequality, ensures that
$S(\psi^{A}) - L < S(\psi_{j}^{A})$. 

We will now construct the universal catalyst state. 
Recall that a catalytic transformation $\ket{\psi}\rightarrow\ket{\phi}$ is possible if and only if
$S(\psi^A) \geq S(\phi^A) \label{eq:mm_CatalyticGrid}$  ~\cite{Kondra2102.11136}.
Let $F$ be the finite set of pairs of indices $(r,s)$ such that
$S(\psi_{r}^A) \geq S(\phi_{s}^A)$.
For $(r,s) \in F$, let
$\tau_{ r, s}$ be a catalyst that enables the transition $\ket{\psi_{r}}\rightarrow\ket{\phi_{s}}$ with an error at most $\frac{\varepsilon}{2}$.
We define:
\begin{equation}
    \tau_{\varepsilon} = \bigotimes_{(r, s) \in F} \tau_{r, s}. \label{eq:mm_UniversalCatalystState}
\end{equation}

As we will now see, the state $\tau_{\varepsilon}$ is the desired universal catalyst state.
Let $\ket{\psi}$, $\ket{\phi}$ be states such that
$S(\psi^A) \geq S(\phi^A)$.
According to the above discussion, we can find indices $r,s$ for which
\begin{subequations} 
\label{eq:mm_UniversalGridSpecific}
\begin{align}
\left\Vert \, \ket{\phi_{s}}\!\bra{\phi_{s}}-\ket{\phi}\!\bra{\phi} \,\right\Vert _{1}   < \frac{\varepsilon}{10},\\
\left\Vert \,\ket{\psi_{r}}\!\bra{\psi_{r}}-\ket{\psi}\!\bra{\psi} \,\right\Vert _{1}   < \frac{\varepsilon}{4},\\
S(\psi^{A}) - L < S(\psi_{r}^{A}),
\end{align}
\end{subequations}
and the conversion $\ket{\psi} \rightarrow \ket{\psi_{r}}$ is possible via LOCC.

If $S(\phi_{s}^{A}) = 0$, then $\ket{\phi_s}$ is a product state, and there is an LOCC transformation $\ket{\psi_{r}}\rightarrow\ket{\phi_{s}}$.
Suppose $S(\phi_{s}^{A}) > 0$.
There exists a state $\ket{\theta}$ such that
$\left\Vert \, \ket{\phi_{s}}\!\bra{\phi_{s}}-\ket{\theta}\!\bra{\theta} \,\right\Vert _{1}  < \frac{\varepsilon}{20}$,
$S(\theta^A)<S(\phi_{s}^A)$.
Hence, there exists another state $\ket{\phi_t}$ from the finite collection defined above, such that
$\left\Vert \, \ket{\phi_{t}}\!\bra{\phi_{t}}-\ket{\theta}\!\bra{\theta} \, \right\Vert _{1}  < \frac{\varepsilon}{10}$,
and $S(\phi_{t}^A)\leq S(\theta^A)$.
Clearly 
$\left\Vert \, \ket{\phi_{t}}\!\bra{\phi_{t}}-\ket{\phi}\!\bra{\phi} \, \right\Vert _{1}  < \frac{\varepsilon}{10} + \frac{\varepsilon}{20} + \frac{\varepsilon}{10} = \frac{\varepsilon}{4}$.
If $S(\phi_{t}^{A}) = 0$, then
$\ket{\phi_t}$ is a product state which can be prepared via LOCC,
just as before.  If $S(\phi_{t}^{A})>0$,
then because $\ket{\phi_s} \neq \ket{\phi_t}$ 
we have that 
\begin{multline}
    S(\phi_{t}^{A}) \leq S(\phi_{s}^{A}) - L \leq  
    S(\phi^{A}) - L \leq S(\psi^{A}) - L \\
    \leq  S(\psi_{r}^{A}),
\end{multline}
where the second inequality follows from the majorization between $\phi^{A}$ and $\phi^{A}_s$.
In any case,
there exists an LOCC protocol that converts 
$\ket{\psi_r}$ into $\ket{\phi_t}$ (or into $\ket{\phi_s}$), utilising the universal catalyst state $\tau_{\varepsilon}$ with an error at most $\frac{\varepsilon}{2}$.  
Since neither $\ket{\psi_r}$ nor $\ket{\phi_{t}}$ (as well as $\ket{\phi_s}$) are
further than $\frac{\varepsilon}{4}$ from respectively $\ket{\psi}$ or $\ket{\phi}$, 
an additional application of the triangle inequality 
concludes the proof.
\end{proof}

The theorem just presented shows that universal catalysis is possible for transformations between bipartite pure states, solving a problem pointed out in~\cite{Kondra2102.11136}. It remains an open question whether universal catalysis is possible also for transformations between mixed states, i.e., whether Eqs. (\ref{eq:UniversalCatalysis}) can be extended to mixed states $\rho^{AB}$ and $\sigma^{AB}$ whenever the conversion $\rho^{AB} \rightarrow \sigma^{AB}$ is possible via catalytic LOCC with a catalyst which depends on the initial and the final state. A similar open question concerns catalytic transformations of entangled states in multipartite settings. We also note that universal catalysis has been previously investigated in quantum thermodynamics, where it was shown that almost all states can be used as universal catalysts~\cite{Lipka-Bartosik2006.16290}.

\section{Catalysis in asymptotic setups}

\subsection{Catalysis for iid systems}

As has been discussed in~\cite{Kondra2102.11136,Lipka-Bartosik2102.11846}, catalytic conversion of entangled states is closely related to asymptotic conversion. For the latter, the goal is to convert $n$ copies of an initial state $\rho$ into $m$ copies of a final state $\sigma$ by using LOCC, allowing for an error which vanishes in the limit of large $n$. The figure of merit in this setup is the optimal conversion rate, i.e., the maximal possible value of $m/n$. Whenever there exists an asymptotic conversion $\rho \rightarrow \sigma$ with unit rate, there also exists a catalytic LOCC protocol converting from $\rho$ into $\sigma$. We refer to the Refs.~\cite{Kondra2102.11136,Lipka-Bartosik2102.11846} for more details.

We will now go one step further, investigating asymptotic conversion of entangled states in the presence of catalysts. To keep the notation simple, we denote the system of Alice and Bob by $S$, and the catalyst will be denoted by $C$. We say that a state $\rho$ can be converted into another state $\sigma$ via \emph{asymptotic catalysis} with rate $R$, if there exists a sequence of catalyst states $\{\tau^{C}_n\}$ and a sequence of LOCC protocols $\{\Lambda_n\}$, such that
\begin{align}
\lim \limits_{n\rightarrow \infty} \left\Vert \Lambda_n \left[\rho^{\otimes n}\otimes\tau^{C}_n \right]-\sigma^{\otimes \lfloor n R \rfloor}\otimes\tau^{C}_n\right\Vert _{1} & = 0,
\end{align}
and
\begin{align}
\mbox{Tr}_{S^{\otimes \lfloor n R \rfloor}}\left\{\Lambda_n\left[\rho^{\otimes n}\otimes\tau^{C}_n \right]\right\} & =\tau^{C}_n
\end{align}
for every integer $n$, where $\Lambda_{n}[\rho^{\otimes n}\otimes\tau_{n}^{C}]$ is
a state in $S^{\otimes\left\lfloor nR\right\rfloor }\otimes C$.
The supremum taken over all such numbers $R$ is the \emph{optimal rate} of converting $\rho$ into $\sigma$ via asymptotic catalysis. With this in the following theorem we prove that in the presence of entangled catalysts asymptotic transformations and single-copy transformations are equally powerful.
\begin{theorem} \label{thm:AsymptoticCatalysis}
A state $\rho$ can be transformed into $\sigma$ via asymptotic catalysis with unit rate if and only if there exists a catalytic LOCC transformation converting one copy of $\rho$ into $\sigma$.
\end{theorem}
\noindent We refer to Appendix \ref{appendix:asymptotic_catalysis} for the proof. We point out that Theorem~\ref{thm:AsymptoticCatalysis} generalizes Theorem 1 of~\cite{Kondra2102.11136}, and also holds for multipartite systems and multipartite LOCC protocols. In this general setup, the existence of an asymptotic catalytic transition ensures the existence of a catalytic transformation on the single-copy level. Note that this kind of transformation has been studied before for pure bipartite states in \cite{DuanPhysRevA.71.042319}, where it has been shown that multiple-copy exact catalytic transformation is equivalent to the single-copy catalytic transformation. On the contrary, our result is more general since it takes into account arbitrarily small error in the final state of the system as well as arbitrarily small correlation in the final state of the system and the catalyst.

\subsection{Asymptotic settings beyond iid}
In the usual independent and identically distributed (iid) scenario, one assumes that Alice and Bob have access to many copies of a bipartite quantum state $\ket{\psi}$, i.e., the total state is given by $\ket{\psi}^{\otimes n}$. We will now go one step further, assuming that Alice and Bob have access to states of the form $\otimes_{i=1}^n\ket{\psi_i}$, where the bipartite states $\ket{\psi_i}$ are not necessarily identical. As we will show below, catalysis provides significant advantage for such non-iid settings. 

In the following, we say that the sequence $\{\ket{\psi_i}\}$ allows to extract a singlet with fidelity $f$ and probability $p$ if there exists an integer $n$ and a probabilistic LOCC protocol $\Lambda$ such that 
\begin{align}
\Tr\left(\Lambda\left[\otimes_{i=1}^{n}\psi_{i}\right]\right)  =p,\\
\frac{\braket{\phi_{2}^{+}|\Lambda\left[\otimes_{i=1}^{n}\psi_{i}\right]|\phi_{2}^{+}}}{p}  =f.
\end{align}
Moreover, we say that the sequence $\{\ket{\psi_i}\}$ allows to extract $m$ singlets with catalysis, if for any $\varepsilon > 0$ there exists an integer $n$, a catalyst state $\tau$, and an LOCC protocol $\Lambda$ such that  
\begin{align}
\left\Vert \Lambda\left[\left(\otimes_{i=1}^{n}\psi_{i}\right)\otimes\tau\right]-\ket{\phi_{2}^{+}}\!\bra{\phi_{2}^{+}}^{\otimes m}\otimes\tau\right\Vert _{1}  <\varepsilon,\\
\Tr_{S}\left\{ \Lambda\left[\left(\otimes_{i=1}^{n}\psi_{i}\right)\otimes\tau\right]\right\}   =\tau.
\end{align}
We are now ready to prove the following theorem.

\begin{theorem} \label{thm:CatalysisBeyondIID}
Given fidelity $f>\frac{1}{2}$ and any $\varepsilon>0$, 
there is a sequence of two-qubit states $\{\ket{\psi_i}\}$ such that for any $n \geq 1$
the probability of converting the state $\otimes_{i=1}^n\ket{\psi_i}$
into a singlet, via LOCC with fidelity $f$, is smaller than $\varepsilon$.
At the same time, an unbounded number of singlets can be extracted with certainty from $\otimes_{i=1}^n\ket{\psi_i}$ with the help of catalysis, as $n\rightarrow \infty$.
\end{theorem}

\begin{proof}
Let $\lambda^{\Psi}_n$ denote the square of the largest Schmidt coefficient of 
the yet to be constructed state $\ket{\Psi_n} = \otimes_{i=1}^n\ket{\psi_i}$.
According to \cite{kondra2021stochastic}, Theorem 3 and below,
the maximal probability of converting the state $\ket{\Psi_n}$
into a singlet $\ket{\phi_2^{+}}$ with fidelity $f$ is
given by
\begin{equation}
    P_{f} = \frac{1 - \lambda^{\Psi}_n}{\sin^{2} \left( \frac{\pi}{4} - \cos^{-1} \sqrt{f} \right)},
\end{equation}
as long as the following quantity is negative:
\begin{equation}
m_n = \sin^{-1}\sqrt{1-\lambda^{\Psi}_n} + \cos^{-1} \sqrt{f} - \frac{\pi}{4} <0.    
\end{equation}
Because $f> \frac{1}{2}$, we have that $\cos^{-1} \sqrt{f} < \frac{\pi}{4}$.
Let us choose $0<\delta<\frac{1}{2}$ such that 
\begin{subequations}
\begin{align}
    \sin^{-1} \sqrt{\delta} + \cos^{-1} \sqrt{f} - \frac{\pi}{4} < 0,\,\,\mbox{and} \\ 
    \frac{\delta}{\sin^{2} \left( \frac{\pi}{4} - \cos^{-1} \sqrt{f} \right)} < \varepsilon.
\end{align}
\end{subequations}
Our goal is to construct $\ket{\Psi_n} = \otimes_{i=1}^n\ket{\psi_i}$
such that $\lambda^{\Psi}_n > 1-\delta$,
and at the same time $\sum_{i=1}^{\infty} S(\psi^{A}_i) = \infty$.

Let $\ket{\psi_i} = \sqrt{1 - p_i} \ket{00} + \sqrt{p_i} \ket{11}$,
where the sequence $0<p_i<1$ is such that
\begin{subequations}
\begin{align}
    \label{eq:infProd}
    \prod \limits_{i=1}^{\infty} (1-p_i) > 1-\delta,  \\
    \label{eq:divseriesEntropies}
   - \sum \limits_{i=1}^{\infty}  p_i \log_2 p_i = \infty.
\end{align}
\end{subequations}
In order to construct such a sequence, we recall that
the series of positive numbers:
$\sum_{k=2}^{\infty} k^{-1} (\log_2 k)^{-(1+u)}$
converges for $u > 0$ and is divergent for $u\leq 0$.
Let us fix $0 < u \leq 1$,
and let $r_1 = \frac{1}{2}$ and
\begin{equation}
    r_k = \frac{1}{k (\log_2 k)^{1+u}},
\end{equation}
for $k=2,3,\ldots$.
Note that if $(a_k)$ is a sequence such that
$\sum_{k=1}^{\infty} |a_k| < \infty$ 
and $a_k>-1$,
then $\sum_{k=1}^{\infty} |\log_2 (1+ a_k)| < \infty$, 
which follows simply from the fact that
$|\log_2 (1+ a_k)| \leq 2(\ln 2)|a_k|$ as long as $|a_k|<\frac{1}{2}$.
By setting $a_k =  - r_k$, 
we get $\sum_{k=1}^{\infty} |\log_2 (1- r_k)| < \infty$,
i.e. $\prod_{k=1}^{\infty} (1-r_k) = C > 0$.
Because the product converges,
there is a number $N$ such that
$\prod_{k=N}^{\infty} (1-r_k) > 1-\delta$.
We set $p_i = r_{i+N}$ and obtain Eq.~\eqref{eq:infProd}.

Eq.~\eqref{eq:divseriesEntropies} follows from a straightforward calculation
(see also \cite{Baccetti_2013}).  
For $k\geq 2$, we have
\begin{eqnarray}
    - r_k \log_2 r_k &=& 
        \frac{1}{k (\log_2 k)^{1+u}} \log_2 \left( k (\log_2 k)^{(1+u)} \right)  \nonumber\\
     &=&\frac{1}{k\, (\log_2 k)^u} + \frac{\log_2 (\log_2 k)^{(1+u)}}{k\, (\log_2 k)^{(1+u)}},
\end{eqnarray}
and the first term makes up a divergent series for $u \leq 1$.
Thus
\begin{equation}
\label{eq:rkDiverges}
    -\sum \limits_{k=2}^{\infty} r_k \log_2 r_k \geq 
        \sum \limits_{k=2}^{\infty}  \frac{1}{ k\, (\log_2 k)^u} = \infty.
\end{equation}
Of course, Eq.~\eqref{eq:rkDiverges} implies Eq.~\eqref{eq:divseriesEntropies}.

Now, because $\delta<\frac{1}{2}$, each $p_i < \frac{1}{2}$, and the square of the largest Schmidt coefficient of the state $\ket{\Psi_n}$ must be 
$\lambda^{\Psi}_n = \prod_{i=1}^{n} ( 1 - p_i) > 1-\delta$.
At the same time, because of Eq.~\eqref{eq:divseriesEntropies}, 
we have that $\sum_{i=1}^{n} S(\psi^{A}_i) = \infty$. Finally, as the entanglement entropy of $\ket{\Psi_n}$ diverges, we can extract an unbounded number of singlets with the help of a catalyst \cite{Kondra2102.11136}. This completes the proof.
\end{proof}

The above theorem demonstrates an advantage of catalysis in settings beyond iid. There exist sequences of two-qubit states $\{\ket{\psi_i}\}$ which do not allow for the extraction of singlets via LOCC without catalysis, as the probability to obtain a singlet (or even any state close to a singlet) is vanishingly small. At the same time, if the remote parties have access to an entangled catalyst, they can in principle extract an unbounded number of singlets from the sequence with certainty. We note that entanglement distillation from non-iid sequences of states has also been discussed in~\cite{Bowen4567558,Buscemi1.3483717,WaeldchenPhysRevLett.116.020502}.

\section{Entanglement catalysis for quantum channels} 

In classical information theory an important development is the Shannon noisy channel coding theorem, where mutual information plays a pivotal role to describe classical capacity of a channel \cite{Shannon_1948,shannon1998mathematical,Covar_1991}. Analogously, in quantum information theory the figure of merit is the quantum capacity~\cite{LloydPhysRevA.55.1613,SchumacherPhysRevA.54.2629,Horodecki_2000}, corresponding to the maximum rate at which a sender can faithfully send qubits to a receiver via a noisy quantum channel. Quantum capacity is closely related to the coherent information of a channel $\Lambda$ and a source state $\rho$~\cite{SchumacherPhysRevA.54.2629,LloydPhysRevA.55.1613}:
\begin{equation}
    I(\rho,\Lambda)=S(\Lambda[\rho])-S(\openone\otimes\Lambda[\ket{\psi_\rho}\!\bra{\psi_{\rho}}]),
\end{equation}
where $\ket{\psi_\rho}$ is some purification of $\rho$. In terms of coherent information, the quantum capacity of the quantum channel $\Lambda$ can be expressed as follows~\cite{LloydPhysRevA.55.1613,Shor_2002,Devetak1377491}:
\begin{equation}\label{definition_quantum_capacity}
    Q(\Lambda)=\lim_{n\rightarrow\infty}\frac{1}{n}\max_{\rho_n}I(\rho_n,\Lambda^{\otimes n}).
\end{equation}
Quantum capacity also coincides with the entanglement generation capacity of a channel \cite{Watrous_2018}. Initially introduced in~\cite{LloydPhysRevA.55.1613}, quantum capacity has been greatly explored for different kinds of channels, such as Pauli channels \cite{Cerf_2000}, bosonic channels~\cite{Holevo_2001, Wolf_2007}, symmetric side channels \cite{Smith_2008}, arbitrarily correlated noise channels \cite{Buscemi_2010}, and low noise channels \cite{Leditzky_2018}. Recently, an experiment has been performed to determine the lower bound to the quantum capacity of two-qubit channels~\cite{Cueves_2017}, which is based on a method to detect lower bounds of quantum capacities \cite{Macchiavello_2016}.
Additionally, LOCC-assisted private and quantum channel capacity has been discussed in \cite{DavisPhysRevA.97.062310}.
For a general overview of the subject we also refer to~\cite{Gyongyosi_2018,Holevo_2020}.

The usual assumption in the study of quantum capacity is that the parties can use many copies of the same channel in parallel. In contrast to this, here we assume that Alice and Bob can use the quantum channel only once and can
have unlimited classical communication. This implies that (in general) they will not be able to send qubits perfectly through the channel. However, as we will now see, the situation changes completely if Alice and Bob can additionally use catalytic LOCC. We assume that Alice and Bob initially share a catalyst\footnote{To simplify the notation, in this section the catalyst system will be denoted by $AB$, and the primary system will be called $A'B'$.} in the state $\tau_n^{AB}$. Alice and Bob are further in possession of registers $A'$ and $B'$ of the same dimension $d_{A'} = d_{B'}$, initialized in the states $\ket{0}\!\bra{0}^{A'}$ and $\ket{0}\!\bra{0}^{B'}$. The aim of the procedure will be to entangle these registers $A'$ and $B'$. For achieving this, Alice and Bob can use a carrier particle $C$ initialized in the state $\ket{0}\!\bra{0}^C$, which can be sent from Alice to Bob via the noisy channel $\Lambda$. The overall initial state is given by 
\begin{equation}
    \sigma_n = \sigma_n^{ABA'B'C} = \tau_n^{AB}\otimes \ket{0}\!\bra{0}^{A'} \otimes \ket{0}\!\bra{0}^{B'} \otimes \ket{0}\!\bra{0}^C. \label{eq:InitialState}
\end{equation}
Consider now a protocol consisting of the following steps.
\begin{enumerate}
\item Preprocessing: Alice and Bob apply an LOCC protocol to $\sigma_n$, with the resulting state $\nu_n$. 
\item Alice uses the channel $\Lambda^C$ to send the carrier particle to Bob. The resulting state is $\chi_n = \Lambda^C[\nu_n]$, where the particle $C$ is now in Bob's lab.
\item Postprocessing: Alice and Bob apply an LOCC protocol to $\chi_n$, resulting in the state $\mu_n$. 
\end{enumerate}

\noindent For the final state $\mu_n$, we require that the state of the catalyst is returned unchanged for each $n$: 
\begin{align}
 & \mathrm{Tr}_{A'B'}\left[\mu_{n}^{ABA'B'}\right]=\tau_{n}^{AB}. \label{eq:Catalyst} 
\end{align}
This requirement is analogous to Eq.~(\ref{eq:CatalyticLOCC1}) for catalytic state transformations, implying that the catalyst can be reused. In the same way as in Eq.~(\ref{eq:CatalyticLOCC2}) we also require that the catalyst decouples from $A'B'$ for large $n$: 
\begin{align}
 \lim_{n\rightarrow\infty}\left\Vert \mu_{n}^{ABA'B'}-\tau_{n}^{AB}\otimes\mu_{n}^{A'B'}\right\Vert _{1}=0. \label{eq:Decoupling}
\end{align}
This requirement limits the amount of correlations between different systems, if they are transformed with the help of the same catalyst.

\begin{figure*}
\centering
\includegraphics[width=1.5\columnwidth]{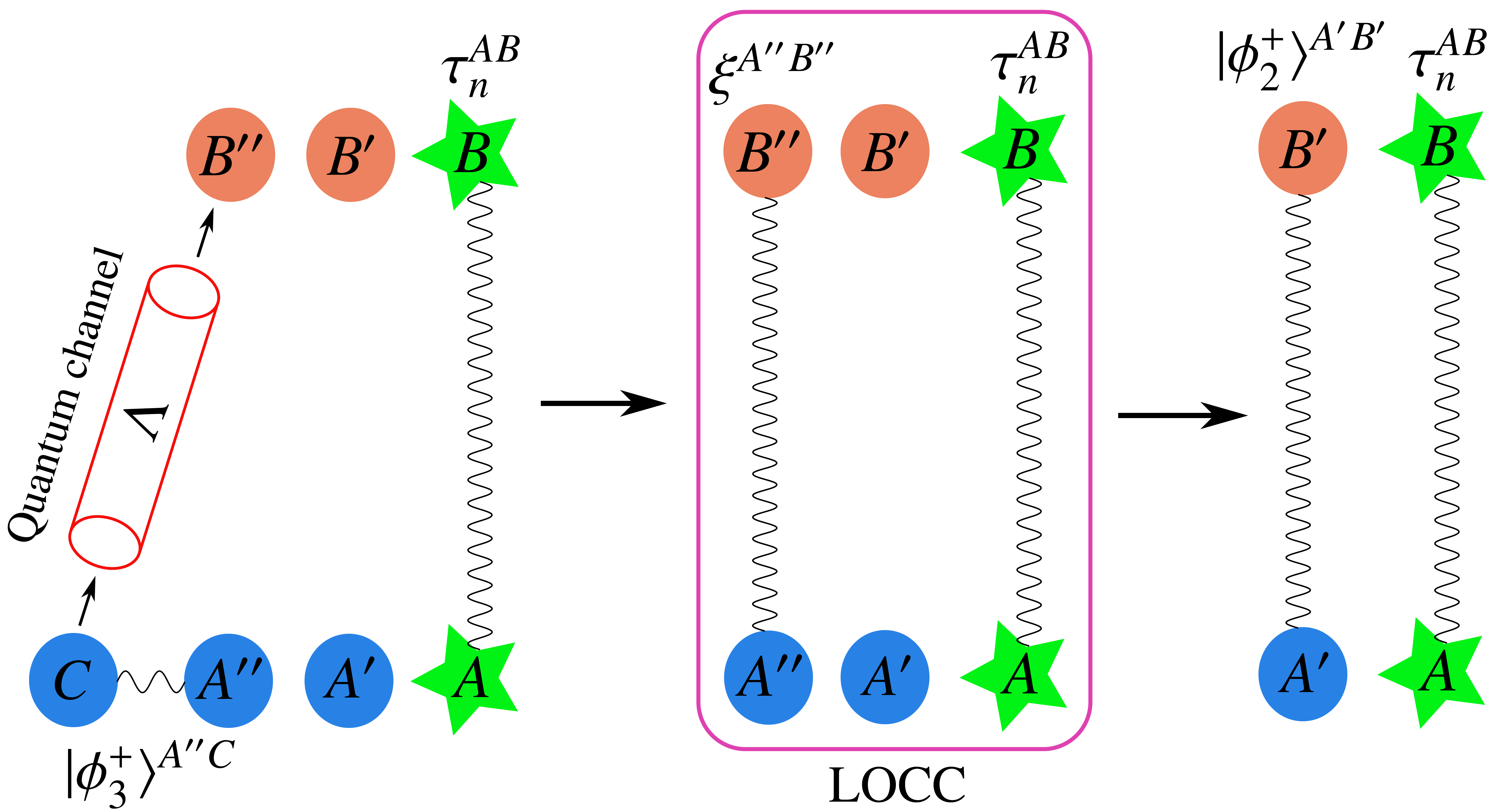}
\caption{Alice prepares a maximally entangled two-qutrit state $\ket{\phi_3^+}^{A''C}$ and sends the particle $C$ to Bob through a noisy channel $\Lambda$, resulting in the state $\xi = \openone \otimes \Lambda (\ket{\phi_3^+}\!\bra{\phi_3^+})$. Additionally, the parties have access to a catalyst in the state $\tau_n^{AB}$. Applying catalytic LOCC Alice and Bob can obtain a singlet when the distillable entanglement of $\xi$ is at least one.}
\label{fig:catalytic_channel}
\end{figure*}

The protocol just presented is the most general procedure which Alice and Bob can apply, if they have access to catalytic LOCC, and a quantum channel $\Lambda$ which can be used only once. A case of special interest arises if the state $\mu_n^{A'B'}$ established in this way is close to a maximally entangled state of $m$ qubits, i.e, 
\begin{equation} \label{eq:ChannelCatalytic}
\lim_{n\rightarrow\infty}\left\Vert \mu_{n}-\tau_{n}^{AB} \otimes \ket{\phi^+_{2^m}}\!\bra{\phi^+_{2^m}}^{A'B'} \right\Vert _{1}=0.
\end{equation}
Since $\ket{\phi^+_{2^m}}$ can be used to teleport $m$ qubits, Alice and Bob can use the procedure to send $m$ qubits with arbitrary accuracy. If Eq.~(\ref{eq:ChannelCatalytic}) is fulfilled for some $m \geq 1$, we say that the channel $\Lambda$ can transmit $m$ qubits.  

Equipped with these tools, we are now ready to define the catalytic capacity $Q_c$ of a quantum channel $\Lambda$ as the maximal number of qubits which the channel can transmit. In more detail, let $\{\tau_n^{AB}\}$ be a suitably chosen sequence of catalyst states, and $\{\mu_n\}$ is a sequence of total final states such that Eqs.~(\ref{eq:Catalyst}) and~(\ref{eq:ChannelCatalytic}) are fulfilled for some $m \geq 1$. The catalytic capacity of $\Lambda$ is the largest possible value of $m$:
\begin{widetext}
\begin{equation}
    Q_{c}(\Lambda)=\max\left\{ m:\lim_{n\rightarrow\infty}\left\Vert \mu_{n}-\tau_{n}^{AB}\otimes\ket{\phi_{2^{m}}^{+}}\!\bra{\phi_{2^{m}}^{+}}^{A'B'}\right\Vert _{1}=0\right\}. \label{eq:CatalyticQuantumCapacity}
\end{equation}
\end{widetext}
If Eq.~(\ref{eq:ChannelCatalytic}) cannot be fulfilled for any $m \geq 1$, we set $Q_c(\Lambda) = 0$. In this context, it is worth noting a recent study \cite{Lipka-Bartosik2102.11846} that looks into the prospect of using catalysis to improve teleportation fidelity. We go into further detail about this in Appendix \ref{appendix_comparison}.

The main differences between the standard quantum capacity and the catalytic capacity are as follows: (i) The definition of standard quantum capacity does not involve any classical communication between the parties, while unlimited classical communication is allowed in the definition of catalytic capacity. (ii) Standard quantum capacity is defined in the limit of large number of copies of the channel, whereas catalytic capacity is defined for a single application of the channel.

We will now consider a concrete example to elucidate the above concept of catalytic capacity, which is also illustrated in Fig.~\ref{fig:catalytic_channel}. We assume that the system $A$ is a qubit, and the carrier particle $C$ is a qutrit. Recalling that the carrier particle is initially in possession of Alice, she can use the preprocessing step to entangle $C$ with an additional qutrit $A''$, creating a maximally entangled two-qutrit state $\ket{\phi_3^+}^{A'' C}$ locally. The particle $C$ is then sent to Bob via the quantum channel $\Lambda$. In this way, Alice and Bob end up sharing the noisy two-qutrit state $\xi = \openone \otimes \Lambda (\ket{\phi_3^+}\!\bra{\phi_3^+})$. So far, the procedure did not make any use of the catalyst. In the postprocessing step, Alice and Bob apply catalytic LOCC to convert $\xi$ into a Bell state. This is possible whenever $\xi$ has distillable entanglement larger or equal than one~\cite{Kondra2102.11136}. In particular, Alice and Bob can obtain the state $\mu_n$ such that $\mathrm{Tr}_{A'B'}[\mu_{n}]=\tau_{n}^{AB}$ and additionally
\begin{equation} \label{eq:ChannelCatalytic2}
\lim_{n\rightarrow\infty}\left\Vert \mu_{n}-\tau_{n}^{AB} \otimes \ket{\phi_2^+}\!\bra{\phi_2^+}^{A'B'} \right\Vert _{1}=0.
\end{equation}
Thus, Alice and Bob can establish a state $\mu_n^{A'B'} \approx \ket{\phi_2^+}\!\bra{\phi_2^+}^{A'B'}$.

The procedure just described shows that a noisy qutrit channel $\Lambda$ can transmit a qubit whenever $E_\mathrm d(\openone \otimes \Lambda [\ket{\phi^+_3}\!\bra{\phi^+_3}]) \geq 1$, where $E_\mathrm d$ is the distillable entanglement. Note that the use of a maximally entangled state is not crucial in this procedure. In fact, instead of creating the two-qutrit state $\ket{\phi^+_3}$, Alice can locally create any other two-qutrit state $\rho$. The procedure will work in the same way, as long as the distillable entanglement of $\openone \otimes \Lambda [\rho]$ is at least one. By the same arguments, the channel $\Lambda$ can transmit $m$ qubits whenever there exists a bipartite state $\rho$ such that \cite{ganardi2023catalytic,Lipka-Bartosik2102.11846}
\begin{equation}
    E_\mathrm d (\openone \otimes \Lambda [\rho]) \geq m. \label{eq:EdM}
\end{equation}

Recall now that the distillable entanglement is bounded below as follows~\cite{Devetak2005}: $E_\mathrm d (\rho^{AB}) \geq S(\rho^A) - S(\rho^{AB})$. Together with Eq.~(\ref{eq:EdM}) we see that a quantum channel can transmit $m$ qubits if the following condition is fulfilled:
\begin{equation}
    S\left(\openone \otimes \Lambda \left[\ket{\phi^+_d}\!\bra{\phi^+_d}\right]\right) \leq \log_2 d - m. \label{eq:MaximallyEntangledCondition}
\end{equation}
Here, $d$ is the dimension of the Hilbert space on which $\Lambda$ is acting. Assuming again that $d=3$, we see that all qutrit channels which are not ``too noisy'' can transmit a qubit as long as $S(\openone\otimes\Lambda[\ket{\phi_{3}^{+}}\!\bra{\phi_{3}^{+}}])\leq0.58$. It is important to note that in this example (and the examples later), for the purpose of simplicity, we regard the pre-processing stage to be Alice preparing a suitable state. However, as previously introduced, the pre-processing stage can generally involve an LOCC protocol between Alice and Bob.

We will now discuss the converse, providing methods to show when a channel cannot send a qubit perfectly. As we will see, a channel can transmit a qubit only if it can transmit one unit of entanglement. As a quantifier of entanglement we use the squashed entanglement defined as~\cite{Christandl_2004}
\begin{multline}
    E_\mathrm{sq}(\rho^{AB}) = \\ \inf \left\{ \frac{1}{2} I(A;B|E):\rho^{ABE} \mathrm{\,\,extension\,\,of\,\,} \rho^{AB} \right\},
\end{multline}
with the quantum conditional mutual information $I(A;B|E) = S(\rho^{AE}) + S(\rho^{BE}) - S(\rho^{ABE})-S(\rho^E)$. For a quantum channel $\Lambda$ we now define the amount of transmitted entanglement as follows:
\begin{multline}
    \Delta E_\mathrm{sq} (\Lambda) = \\ \sup_{\rho^{ABC}} \left\{ E^{A|BC}_\mathrm{sq}(\Lambda^C[\rho^{ABC}]) - E^{AC|B}_\mathrm{sq}(\rho^{ABC}) \right\}, \label{eq:TransmittedEntanglement}
\end{multline}
and the supremum is taken over all tripartite states $\rho^{ABC}$. Equipped with these tools, we are now ready to prove the following theorem.
\begin{theorem} \label{thm:Esq}
The catalytic capacity of a channel $\Lambda$ is bounded above as
\begin{equation}
    Q_c(\Lambda) \leq \Delta E_\mathrm{sq} (\Lambda).
\end{equation}
\end{theorem}

\begin{proof}
    We allow Alice and Bob to use the most general communication protocol as described below Eq.~(\ref{eq:InitialState}). In the preprocessing step, Alice and Bob apply LOCC to the state $\sigma_n$. In a general LOCC protocol Alice and Bob can also attach local systems. We thus denote the state after the preprocessing by $\nu_n = \nu_n^{\tilde A \tilde B C}$, where $\tilde A$ includes $AA'$, and all additional particles that Alice has attached locally, and similar for $\tilde B$. Noting that the squashed entanglement does not increase under LOCC~\cite{Christandl_2004}, we have 
\begin{equation}
    E_\mathrm{sq}^{\tilde A C |\tilde B}(\nu_n) \leq E_\mathrm{sq}^{AA'C |BB'}(\sigma_n).
\end{equation}
In the second step of the protocol, Alice sends the carrier particle $C$ to Bob, using the quantum channel $\Lambda^C$. The resulting state is $\chi_n^{\tilde A \tilde B C} = \Lambda^C[\nu_n^{\tilde A \tilde B C}]$. Note that the squashed entanglement between Alice and Bob can increase in this process, and the increase is bounded by $\Delta E_\mathrm{sq} (\Lambda)$:
\begin{equation}
    E_\mathrm{sq}^{\tilde A | \tilde B C} (\chi_n) - E_\mathrm{sq}^{\tilde A C | \tilde B} (\nu_n) \leq \Delta E_\mathrm{sq} (\Lambda).
\end{equation}
In the postprocessing step, an LOCC protocol is applied to $\chi_n^{\tilde A \tilde B C}$. In this step, Alice and Bob also discard all particles apart from $AA'BB'$, and the final state is denoted by $\mu_n = \mu_n^{AA'BB'}$. Again using the fact that the squashed entanglement does not increase under LOCC we have
\begin{equation}
    E_\mathrm{sq}^{AA'|BB'}(\mu_n) \leq E_\mathrm{sq}^{\tilde A|\tilde B C} (\chi_n).
\end{equation}
Combining these arguments, we see that the overall increase of entanglement in this procedure is bounded as follows: 
\begin{align}
&E_\mathrm{sq}^{AA'|BB'}(\mu_{n})-E_\mathrm{sq}^{AA'C|BB'}(\sigma_{n})  \nonumber\\ &\leq E_\mathrm{sq}^{\tilde{A}|\tilde{B}C}(\chi_{n})-E_\mathrm{sq}^{\tilde{A}C|\tilde{B}}(\nu_{n}) 
  \leq\Delta E_\mathrm{sq}(\Lambda). \label{eq:ConverseProof}
\end{align}

Recall now that $\mu_n^{AB} = \tau_n^{AB}$ is the state of the catalyst, which does not change in this procedure. Moreover, by the properties of the squashed entanglement~\cite{Christandl_2004} we have
\begin{align}
&E_\mathrm{sq}^{AA'|BB'}(\mu_{n}^{AA'BB'})  \geq E_\mathrm{sq}^{A|B}(\mu_{n}^{AB})+E_\mathrm{sq}^{A'|B'}(\mu_{n}^{A'B'})\nonumber\\   & \qquad\qquad\qquad\qquad\,=E_\mathrm{sq}^{A|B}(\tau_{n}^{AB})+E_\mathrm{sq}^{A'|B'}(\mu_{n}^{A'B'}),\\
&E_\mathrm{sq}^{AA'C|BB'}(\sigma_{n})  =E_\mathrm{sq}^{A|B}(\tau_n^{AB}).
\end{align}
Together with Eq.~(\ref{eq:ConverseProof}) we obtain
\begin{equation}
    E_\mathrm{sq}^{A'|B'}(\mu_n^{A'B'}) \leq \Delta E_\mathrm{sq}(\Lambda).
\end{equation}

If the channel $\Lambda$ can transmit $m$ qubits, then the state $\mu_n^{A'B'}$ can be made arbitrarily close to $\ket{\phi^+_{2^m}}$. Using the continuity of the squashed entanglement~\cite{Alicki_2004} together with $E_\mathrm{sq}(\ket{\phi^+_{2^m}})=m$ we obtain $m \leq \Delta E_\mathrm{sq}(\Lambda)$, and the proof is complete.
\end{proof}

An immediate consequence of Theorem~\ref{thm:Esq} is that entanglement breaking channels have zero catalytic capacity. Since the action of any entanglement breaking channel can be simulated by LOCC~\cite{HorodeckiEntanglementBreakingChannels}, it follows that $\Delta E_{\mathrm{sq}}(\Lambda)=0$ whenever $\Lambda$ is entanglement breaking. Moreover, since catalytic capacity is an integer, it follows that any channel with $\Delta E_{\mathrm{sq}}(\Lambda)<1$ has zero catalytic capacity. This also holds for channels which -- in principle -- can establish entanglement, but which cannot establish a singlet, even when entangled catalysts are used.

\section{Applications for entanglement catalysis of noisy quantum channels}

While the evaluation of the transmitted squashed entanglement $\Delta E_\mathrm{sq}$ is very challenging in general, in this section we will see that for certain types of quantum channels it is possible to obtain computable upper bounds, thus leading to simple upper bounds on the catalytic capacity. 

A general quantum channel $\Lambda$ acts on one part of a bipartite quantum state $\rho$ as follows:
\begin{equation}
\openone\otimes\Lambda(\rho)=\sum_{i=1}^{k}(\openone\otimes K_{i})\rho(\openone\otimes K_{i}^{\dagger}),
\end{equation}
where $\sum_{i=1}^k K_i^\dagger K_i=\openone$ and $k$ is the minimal number of Kraus operators. The final state  $\openone\otimes\Lambda(\rho)$ can also be expressed as 
\begin{equation}
    \openone\otimes\Lambda(\rho)=\sum_i p_i \sigma_i,
\end{equation}
with probabilities $p_i$ and states $\sigma_i$ given as
\begin{subequations}
\begin{align}
p_{i} & =\mbox{Tr}\left[(\openone\otimes K_{i})\rho(\openone\otimes K_{i}^{\dagger})\right],\\
\sigma_{i} & \ensuremath{=\frac{(\openone\otimes K_{i})\rho(\openone\otimes K_{i}^{\dagger})}{p_{i}}}.
\end{align}
\end{subequations}
Assume now that the initial state $\rho$ is maximally entangled, i.e., $\rho = \ket{\phi^+_d}\!\bra{\phi^+_d}$, where $d$ is the dimension of the Hilbert space on which $\Lambda$ is acting. All states $\sigma_i$ are pure in this case, and it holds that 
\begin{equation}
    S(\openone\otimes\Lambda[\ket{\phi^+_d}\!\bra{\phi^+_d}]) \leq H(p_i).
\end{equation}
with $H(p_{i})=-\sum p_{i}\log_{2}p_{i}$. Recall that the channel $\Lambda$ can transmit $m$ qubits if Eq.~(\ref{eq:MaximallyEntangledCondition}) is fulfilled. This means that the channel can transmit $m$ qubits whenever 
\begin{equation}
    H(p_i) \leq \log_2d-m. \label{eq:Kraus}
\end{equation}
From these results we obtain a lower bound on the catalytic capacity:
\begin{equation}
Q_{c}(\Lambda)\geq\left\lfloor \log_{2}d-H(p_{i})\right\rfloor.
\end{equation}
As a consequence, any quantum channel $\Lambda$ of dimension $d \geq 4$ can transmit at least one qubit if the channel can be decomposed into (at most) $2$ Kraus operators.

As an example, consider a quantum channel of the form 
\begin{equation}
    \Lambda[\rho] = (1-p) \rho + p U \rho U^\dagger
\end{equation}
with some unitary $U$, and $p\in[0,1/2]$. From the above discussion it follows that $\Lambda$ can transmit $m$ qubits when the following inequality is fulfilled:
\begin{equation}
    h(p) \leq \log_2 d - m. \label{eq:Example1}
\end{equation}
For $d=3$ and $m=1$ this condition is fulfilled for $p\in[0,0.1403]$. As noted above, for $d \geq 4$ and $m=1$ the inequality in Eq.~(\ref{eq:Example1}) is always true, which means that any such channel can perfectly transmit a single qubit. 

So far we have shown that various quantum channels can perfectly transmit qubits in the presence of entangled catalysts. We will now demonstrate how to use the converse in Theorem \ref{thm:Esq}, which allows to establish upper bounds on the catalytic capacity. For this,
we need to evaluate the transmitted entanglement $\Delta E_\mathrm{sq}$, which is very challenging in general. However, for some channels the expression in Eq.~(\ref{eq:TransmittedEntanglement}) can be simplified significantly. Note that $\Delta E_\mathrm{sq}$ quantifies how much entanglement can be distributed via a given quantum channel $\Lambda$. According to Theorem 8 in~\cite{PhysRevA.92.012335}, for a single-qubit Pauli channel 
\begin{equation}
    \Lambda_p[\rho] = \sum_{i=0}^3 p_i \sigma_i \rho \sigma_i
\end{equation}
the optimal way to distribute entanglement is to send one half of a Bell state through the channel: 
\begin{equation}
    \Delta E_\mathrm{sq}\left(\Lambda_{p}\right)=E_\mathrm{sq}\left(\openone \otimes \Lambda_{p}\left[\ket{\phi_{2}^{+}}\!\bra{\phi_{2}^{+}}\right]\right).
\end{equation}
Similarly, if the quantum channel is a tensor product of (possibly different) Pauli channels, the optimal entanglement distribution procedure is again to send one half of a maximally entangled state. In particular, for a two-qubit channel of the form $\Lambda_p \otimes \Lambda_p$ we have 
\begin{equation}
    \Delta E_\mathrm{sq}\left(\Lambda_{p}\otimes\Lambda_{p}\right)=2E_\mathrm{sq}\left(\openone \otimes \Lambda_{p}\left[\ket{\phi_{2}^{+}}\!\bra{\phi_{2}^{+}}\right]\right). \label{eq:Pauli}
\end{equation}

Note now that the squashed entanglement is bounded above by the entanglement of formation $E_f$~\cite{Christandl_2004}. For pure states $E_f$ is defined as the entanglement entropy: $E_f(\ket{\psi}^{AB}) = S(\psi^A)$. For mixed states, $E_f$ corresponds to the minimal average entanglement of the state~\cite{BennettPhysRevA.54.3824}: $E_f(\rho) = \min \sum_i p_i E_f(\ket{\psi_i})$, where the minimum is taken over all pure state decompositions of the state $\rho$. From Eq.~(\ref{eq:Pauli}) and the fact that $E_\mathrm{sq}$ is bounded from above by the entanglement of formation we obtain
\begin{equation}\label{squashed_formation}
    \Delta E_\mathrm{sq}\left(\Lambda_{p}\otimes\Lambda_{p}\right) \leq 2E_f\left(\openone \otimes \Lambda_{p}\left[\ket{\phi_{2}^{+}}\!\bra{\phi_{2}^{+}}\right]\right).
\end{equation}
Thus, from Theorem~\ref{thm:Esq} we see that the two-qubit channel $\Lambda_p \otimes \Lambda_p$ has zero catalytic capacity if
\begin{equation}
    E_f\left(\openone \otimes \Lambda_{p}\left[\ket{\phi_{2}^{+}}\!\bra{\phi_{2}^{+}}\right]\right) < \frac{1}{2}. \label{eq:PauliEf}
\end{equation}
Since the entanglement of formation has a closed expression for all states of two qubits~\cite{WoottersPhysRevLett.80.2245}, it is straightforward to check if a given Pauli channel fulfills Eq.~(\ref{eq:PauliEf}). In particular, note that $\openone \otimes \Lambda_{p}\left[\ket{\phi_{2}^{+}}\!\bra{\phi_{2}^{+}}\right]$ is diagonal in the Bell basis, with entanglement of formation given by~\cite{BennettPhysRevA.54.3824}
\begin{multline}\label{expression_formation}
    E_f\left(\openone \otimes \Lambda_{p}\left[\ket{\phi_{2}^{+}}\!\bra{\phi_{2}^{+}}\right]\right) \\= h\left(\frac{1}{2} + \sqrt{p_{\max} (1-p_{\max})}\right)
\end{multline}
for $p_{\max} > 1/2$ and $E_f = 0$ otherwise, where $p_{\max}=\max\{p_i\}$. This means that $\Lambda_p \otimes \Lambda_p$ has zero catalytic capacity for $p_{\max} < 0.813$.

We will now extend our results to $n$ copies of a single-qubit Pauli channel $\Lambda_p$. Using the equality
\begin{multline}
S\left(\openone\otimes\Lambda_{p}^{\otimes n}\left[\ket{\phi_{2^{n}}^{+}}\!\bra{\phi_{2^{n}}^{+}}\right]\right)\\=nS\left(\openone\otimes\Lambda_{p}\left[\ket{\phi_{2}^{+}}\!\bra{\phi_{2}^{+}}\right]\right)=nH(p_{i})
\end{multline}
together with Eq.~(\ref{eq:MaximallyEntangledCondition}), we see that the total channel $\Lambda_p^{\otimes n}$ can transmit $m$ qubits whenever
\begin{equation}
H(p_{i})\leq1-\frac{m}{n}. \label{eq:n_pauli_m}
\end{equation}
As we show in Appendix \ref{sec:maxEntropy}, for a given $p_{\max} =\max \{p_k\}$ the maximum value of $H(p_i)$ is achieved for $p_0=p_{\max}$, $p_1=p_2=p_3=(1-p_{\max})/3$. Therefore, if Eq. (\ref{eq:n_pauli_m}) is satisfied for this distribution, then it holds true for any probability distribution with the same $p_{\max}$. In Fig. \ref{fig:max_p_vs_n}, we show values of $p_{\max}$ which allow to send a single qubit perfectly as a function of $n$.  As an example, for $n=2$, we can send a qubit perfectly when $p_{\max}>0.926$, resulting in $Q_c(\Lambda_p\otimes\Lambda_p)=1$ for $0.926<p_{\max}<1$. Moreover, our results imply that for any $p_{\max}>0.8107$ there exists some $n$ such that $\Lambda_p^{\otimes n}$ can transmit (at least) one qubit perfectly.

We will now provide converse bounds for $n$ copies of a Pauli channel, showing when the channel cannot send a certain number of qubits perfectly. Using results from~\cite{PhysRevA.92.012335} (in particular Theorem 8 there), we can generalize Eq.~(\ref{eq:Pauli}) as follows:
\begin{equation}
\Delta E_{\mathrm{sq}}\left(\Lambda_{p}^{\otimes n}\right)=nE_{\mathrm{sq}}\left(\openone \otimes \Lambda_{p}\left[\ket{\phi_{2}^{+}}\!\bra{\phi_{2}^{+}}\right]\right).
\end{equation}
From Theorem~\ref{thm:Esq} we see that the channel $\Lambda_p^{\otimes n}$ can transmit $m$ qubits only if 
\begin{equation}
E_{\mathrm{sq}}\left(\openone \otimes \Lambda_{p}\left[\ket{\phi_{2}^{+}}\!\bra{\phi_{2}^{+}}\right]\right)\geq\frac{m}{n}. \label{eq:EsqNpauli}
\end{equation}
Recalling that the squashed entanglement is bounded above by the entanglement of formation, we further obtain
\begin{equation}
E_f\left(\openone \otimes \Lambda_{p}\left[\ket{\phi_{2}^{+}}\!\bra{\phi_{2}^{+}}\right]\right)\geq\frac{m}{n}. \label{eq:EfNpauli}
\end{equation}
If this condition is violated, then it is not possible to transmit $m$ qubits with $n$ copies of $\Lambda_p$. Using Eq.~(\ref{expression_formation}), it is possible to obtain values of $p_{\max}$ for which this condition is violated. In Fig.~\ref{fig:max_p_vs_n} we show the values of $p_{\max}$ for which no qubit can be transmitted by using $n$ copies of $\Lambda_p$, which implies that the catalytic capacity is zero in those cases.

Since entanglement of formation has a closed expression for all states of two qubits~\cite{WoottersPhysRevLett.80.2245}, the condition in Eq.~(\ref{eq:EfNpauli}) is easily verifiable. As we will now see, an improved numerical condition can be obtained from Eq.~(\ref{eq:EsqNpauli}), using a numerical approximation of the squashed entanglement. We will demonstrate this explicitly for a Pauli channel of the form
\begin{equation}
\Lambda[\rho]=p\rho+(1-p)\sigma_{z}\rho\sigma_{z}. \label{eq:PauliNumerics}
\end{equation}
When applied on one half of a maximally entangled state we obtain 
\begin{equation}
\openone\otimes\Lambda[\ket{\phi_{2}^{+}}\!\bra{\phi_{2}^{+}}]=p\ket{\phi_{2}^{+}}\!\bra{\phi_{2}^{+}}+(1-p)\ket{\phi_{2}^{-}}\!\bra{\phi_{2}^{-}}. \label{eq:EsqNumerics}
\end{equation}
A purification of this state can now be defined as follows:
\begin{multline}
\ket{\psi}^{ABE_{1}E_{2}}=\openone^{AB}\otimes U^{E_{1}E_{2}}\Big(\sqrt{p}\ket{\phi_{2}^{+}}\otimes\ket{00}\\+\sqrt{1\!-\!p}\ket{\phi_{2}^{-}}\otimes\ket{01}\Big),
\end{multline}
where $E_1$ and $E_2$ are qubit systems, and $U^{E_1E_2}$ is some two-qubit unitary. The reduced state $\rho^{ABE_1}$ is a valid extension of the state in Eq.~(\ref{eq:EsqNumerics}). A numerical upper bound on the squashed entanglement can now be obtained by randomly generating the two-qubit unitaries $U^{E_1E_2}$ and evaluating $I(A;B|E_1)/2$ for the state $\rho^{ABE_1}$. For a million iterations, we find that for $p=0.817$ (see also Fig. \ref{fig:squashed}) the upper bound on the squashed entanglement is 0.49988. Together with Eq.~(\ref{eq:EsqNpauli}), this means that for $p < 0.817$ two copies of the channel cannot transmit a qubit perfectly. This result should be compared to the bound obtainable from Eq.~(\ref{eq:EfNpauli}), which gives $p < 0.813$. Moreover, using the arguments presented above it is straightforward to see that for $p \geq 0.89$ two copies of the channel can be used to transmit a qubit perfectly. Summarizing, two copies of the channel in Eq.~(\ref{eq:PauliNumerics}) can transmit a qubit perfectly for $p \geq 0.89$, and cannot send a qubit perfectly when $p < 0.817$. 

\begin{figure}
\includegraphics[width=\columnwidth]{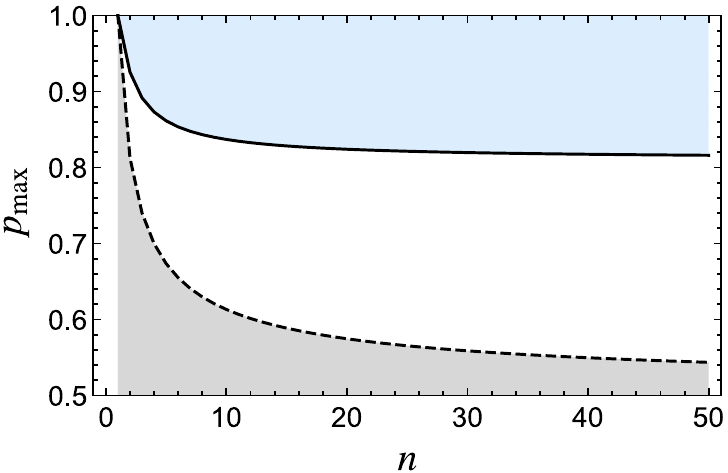}
\caption{Catalytic communication with $n$ Pauli channels. Solid curve shows the values of $p_{\max}$ beyond which the catalytic capacity of $\Lambda_p^{\otimes n}$ is at least one (blue shaded region). Dashed curve shows the values of $p_{\max}$ below which the catalytic capacity of $\Lambda_p^{\otimes n}$ is zero (gray shaded region).}
\label{fig:max_p_vs_n}
\end{figure}

\begin{figure}
\includegraphics[width=\columnwidth]{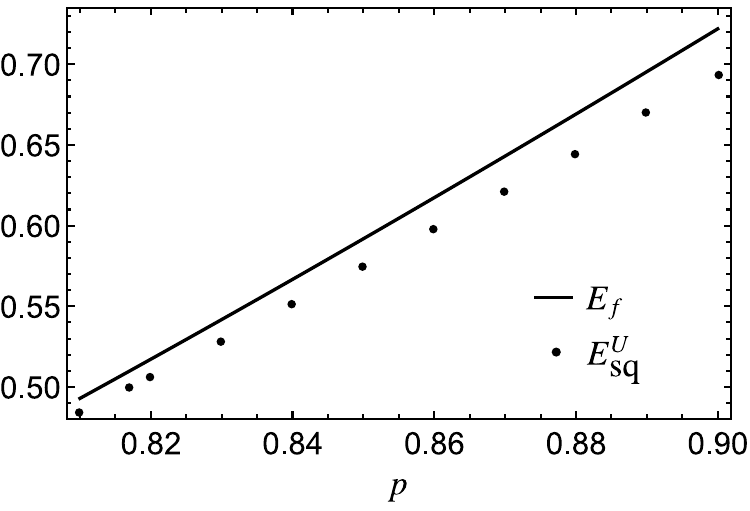}
\caption{Entanglement of formation (solid line) and a numerical upper bound on squashed entanglement (dots) for the state given in Eq. (\ref{eq:EsqNumerics}) as a function of $p$.}
\label{fig:squashed}
\end{figure}

As a final application, we will consider the general amplitude damping channel $\Lambda_{\mathrm{ad}}$ with Kraus operators \cite{DUTTA20162191}
\begin{eqnarray}
   K_0&=&|0\rangle\langle 0|+\sqrt{1-p}\sum_{i=1}^{d-1} |i\rangle\langle i|,\\
   K_m&=&\sqrt{p}|0\rangle\langle m|, \,\, m=\{1,2,\cdots,d-1\}.
\end{eqnarray}
We will now evaluate the number of qubits which the channel can transmit by using Eq.~(\ref{eq:MaximallyEntangledCondition}). Applying the channel onto one half of the maximally entangled state $\ket{\phi^+_d}$ we obtain the state
\begin{align}
&\rho^{AB}=\openone\otimes\Lambda_{\mathrm{ad}}\left(\ket{\phi_{d}^{+}}\!\bra{\phi_{d}^{+}}\right)\nonumber\\
&\quad\,\,\,\,=\frac{\left(d-(d-1)p\right)}{d}|\psi\rangle\langle \psi|+\frac{p}{d}\sum_{m=1}^{d-1}|m0\rangle \langle m0|,
\end{align}
where 
\begin{equation}
    |\psi\rangle=\frac{\sqrt{1-p}}{\sqrt{d-(d-1)p}}\left(\frac{1}{\sqrt{1-p}}|00\rangle+\sum_{m=1}^{d-1}|mm\rangle\right).
\end{equation} 
As $|\psi\rangle$ and $|m0\rangle$ are orthogonal for all $m>0$, the eigenvalues of the state $\rho^{AB}$ are $0$, $\left(d-(d-1)p\right)/d$ and $p/d$ with degeneracy $d(d-1)$, 1 and $(d-1)$ respectively. With this, we can easily calculate $S(\rho^A)-S(\rho^{AB})$ and find the range of $p$ where the channel can transmit $m$ qubits. In Table~\ref{tab:AmplitudeDamingChannel} we show the range of $p$ where $S(\rho^A)-S(\rho^{AB})\geq 1$ for different $d$, implying that a qubit can be send perfectly in this parameter range.


\begin{table}
\begin{center}
\begin{tabular}{| c | c |}
 \hline
$d$ & $p$\\ 
\hline 
3 & < 0.16 \\
\hline 
4 & < 0.25 \\
\hline 
5 & < 0.32 \\
\hline 
6 & < 0.36 \\
\hline 
7 & < 0.40 \\
\hline 
8 & < 0.43 \\
\hline
\end{tabular}
\end{center}
\caption{Parameter range of the general amplitude damping channel allowing for the transmission of one qubit.}\label{tab:AmplitudeDamingChannel}
\end{table}

\section{Enhancing entanglement distribution with catalysis}

Consider a setup where Alice wishes to send one qubit to Bob. For this purpose, Alice has access to quantum channel of length $l$, with an intermediate node able to assist the parties in the procedure, see Fig.~\ref{fig:entanglement_distribution}. We assume that the quantum channel is a single-qubit depolarizing channel 
\begin{equation}
\Lambda_{l}[\rho]=e^{-\alpha l}\rho+(1-e^{-\alpha l})\frac{\openone}{2}
\end{equation}
with a damping parameter $\alpha \geq 0$.
If the intermediate node is not used, the channel can distribute entanglement for $l < \ln 3/\alpha$, and becomes entanglement breaking for $l \geq \ln 3/\alpha$. As we will now see, using the intermediate node in this setup will not change this property: the channel will remain entanglement breaking. For this, consider a subnormalized channel of the form
\begin{equation}
\Lambda'[\rho]=\Lambda_{l-s}\left[K\Lambda_{s}\left(\rho\right)K^{\dagger}\right],
\end{equation}
where $0 \leq s \leq l$ and $K$ is a Kraus operator, i.e., $K^\dagger K \leq \openone$. Moreover, $s$ represents the position of the node, see also Fig.~\ref{fig:entanglement_distribution}. It is straightforward to see that this procedure allows to establish entanglement between Alice and Bob if and only if the state 
\begin{equation}
\sigma=\frac{\openone\otimes\Lambda'[\ket{\phi^{+}}\!\bra{\phi^{+}}]}{\mathrm{Tr(\openone\otimes\Lambda'[\ket{\phi^{+}}\!\bra{\phi^{+}}])}}
\end{equation}
is entangled for some Kraus operator $K$.

\begin{figure}
\includegraphics[width=\columnwidth]{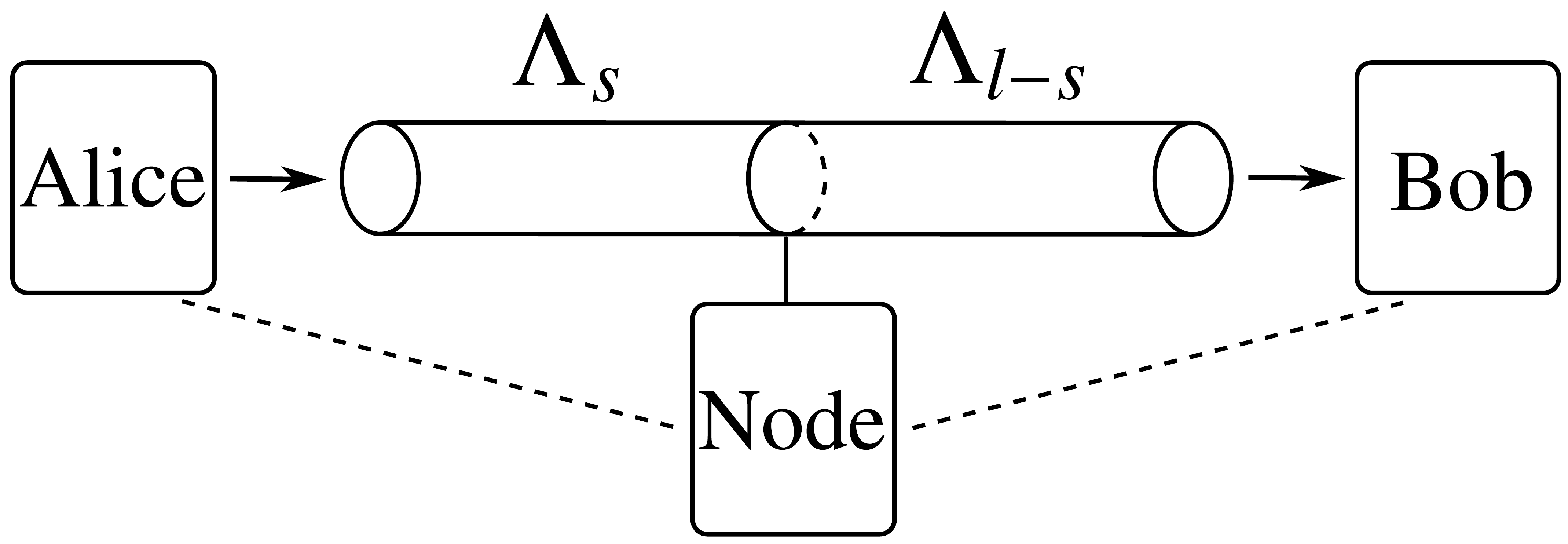}
\caption{Entanglement distribution from Alice to Bob via depolarizing qubit channel of length $l$. Additionally the parties can use an intermediate node having distance $s$ from Alice. The node can perform quantum operations and exchange classical information (dashed lines) with Alice and Bob.}
\label{fig:entanglement_distribution}
\end{figure}

We will now prove that for $l \geq \ln 3/\alpha$ the state $\sigma$ is not entangled for any choice of $K$ and $s$. Note that it is enough to prove the statement for $l = \ln 3/\alpha$. For this, we define
\begin{align}
\ket{\psi}  =\frac{1}{\sqrt{q}}(\openone\otimes K)\ket{\phi^{+}},\\
\mu  =\frac{\left(\openone\otimes K\right)\openone\otimes\Lambda_{s}\left[\phi^{+}\right]\left(\openone\otimes K^{\dagger}\right)}{\mathrm{Tr}\left[\left(\openone\otimes K\right)\openone\otimes\Lambda_{s}\left[\phi^{+}\right]\left(\openone\otimes K^{\dagger}\right)\right]}
\end{align}
with $q = \mathrm{Tr}[KK^\dagger]/2$. Notice that $\psi^B = KK^\dagger/\mathrm{Tr}[KK^\dagger]$. Equipped with these tools we obtain
\begin{align}
\mu = e^{-\alpha s}\ket{\psi}\!\bra{\psi} + (1-e^{-\alpha s})\frac{\openone}{2}\otimes\psi^{B}.
\end{align}
Using these results, we can now express the state $\sigma$ as follows:
\begin{align}
\sigma & =e^{-\alpha l}\ket{\psi}\!\bra{\psi}+e^{-\alpha(l-s)}(1-e^{-\alpha s})\frac{\openone}{2}\otimes\psi^{B}+\nonumber\\
 &\quad\,\, (1-e^{-\alpha(l-s)})e^{-\alpha s}\psi^{A}\otimes\frac{\openone}{2}+\nonumber\\
 &\quad\,\,(1-e^{-\alpha(l-s)})(1-e^{-\alpha s})\frac{\openone}{2}\otimes\frac{\openone}{2}. 
\end{align}
Our goal now is to show that for $l =  \ln 3/\alpha$ the state $\sigma$ is not entangled for any choice of $\ket{\psi}$. Since $\ket{\psi}$ is a two-qubit state, we can set $\ket{\psi} = \cos \beta \ket{00} + \sin \beta \ket{11}$. Recall now that a two-qubit state $\sigma$ is entangled if and only if the determinant of the partially transposed matrix $\sigma^{T_A}$ is negative~\cite{AugusiakPhysRevA.77.030301}. For $l =  \ln 3/\alpha$ and replacing $s$ by $s'/\alpha$, the determinant of $\sigma^{T_A}$ is given by
\begin{align}
    \det (\sigma^{T_A}) = \frac{e^{-4 s'}\cos^2 2\beta}{20736}\Big([f(s')]^2-\nonumber\\    f(s')\left(3+e^{2s'}\right)^2\sin^2 2\beta \Big),
\end{align}
where $f(s')=9-10 e^{2 s'}+e^{4 s'}$. As $f(s')\leq 0$ for $s'=[0,\ln 3]$, $ \det (\sigma^{T_A})$ is nonnegative for all $\beta$. This proves that the setup shown in Fig.~\ref{fig:entanglement_distribution} can distribute entanglement if and only if $l <  \ln 3/\alpha$. 

On the other hand, if entangled catalysts are allowed in the procedure, the setup can distribute entanglement in the range $l < 2\ln 3/\alpha$. To see this, we choose $s = l/2 < \ln 3/\alpha$. Alice first sends one half of a maximally entangled state through the channel, leading to a mixed entangled state between Alice and the node. By using catalytic LOCC between Alice and the node, it is possible to establish a pure entangled state resorting to results in~\cite{Kondra2102.11136} and the fact that all entangled two-qubit states are distillable~\cite{HorodeckiPhysRevLett.78.574}. The particle is then sent through the second half of the channel from the node to Bob, leading to an entangled state between Alice and Bob, as claimed. 

Note that for $l \geq 2\ln 3/\alpha$ the setup cannot distribute entanglement, even if catalysis is used. In this case, it must be either $s \geq \ln 3/\alpha$ or $l - s \geq \ln 3/\alpha$. This means that either the first part of the channel or the second part of the channel is entanglement breaking.

\section{Discussion}

In this article we investigated different aspects of entanglement catalysis, both for quantum states and quantum channels. We proved the existence of universal catalysts, enabling all possible transformations between bipartite pure states. We investigated asymptotic state transitions in the presence of catalysts, proving an equivalence between the asymptotic and the single-copy transformations in the catalytic regime. Moreover, we showed that catalysis offers a significant advantage in asymptotic non-iid setings, where one aims to extract singlets from a sequence of non-identical quantum states. There exist sequences of two-qubit states which allow to extract an unbounded number of singlets if catalysis is used, but the probability to extract even one singlet without catalysis is vanishingly small. 

We further investigated the role of catalysis for noisy quantum channels. Using entangled catalysts, it is possible to turn a noisy quantum channel into a noiseless one, being able to transmit a certain number of qubits with arbitrary precision. We showed that any quantum channel can faithfully transmit qubits as long as the channel is not too noisy. Concretely, any quantum channel of dimension $d$ can transmit $m$ qubits whenever the von Neumann entropy of its Choi state is not larger than $\log_2 d - m$. Remarkably, for $d \geq 4$ any quantum channel can transmit at least one qubit if the Choi state has rank 2. We introduced the catalytic capacity, corresponding to the number of qubits which can be reliably transmitted in the presence of catalysts. We developed tools to estimate the catalytic capacity, allowing to determine the exact value for various quantum channels. We also demonstrated that catalysis is useful for entanglement distribution via single-qubit depolarizing channels. In particular, we showed that a depolarizing channel which is useless for entanglement distribution without catalysis can be turned into a useful channel by introducing a catalyst and an intermediate node which assists the parties with local measurements and classical communication.

While we focused on entanglement catalysis in this article, recent results show that catalysis is useful within the theories of quantum thermodynamics~\cite{GOUR20151,Brandao3275,Lipka-Bartosik2006.16290,Wilminge19060241,BoesPhysRevLett.122.210402,wilming2020entropy,shiraishi2020quantum,Henao2021catalytic,henao2022catalytic} and quantum coherence~\cite{BuPhysRevA.93.042326,StreltsovRevModPhys.89.041003,AbergPhysRevLett.113.150402,Vaccaro_2018,LostaglioPhysRevLett.123.020403,Takagi2106.12592,char2021catalytic,DattaPhysRevLett.130.240204}, giving access to state transformations which are impossible without the catalyst. It is reasonable to believe that several results presented in this work can also be extended to other resource theories, thus significantly enhancing our understanding of state transformations in quantum mechanics.

\section*{Acknowledgements}
We thank Stanis\l{}aw Kurdzia\l{}ek, Rafa\l{} Demkowicz-Dobrza\'nski, and Wojciech G\'orecki for discussion. This work was supported by the ``Quantum Optical Technologies'' project, carried out within the International Research Agendas programme of the Foundation for Polish Science co-financed by the European Union under the European Regional Development Fund, the ``Quantum Coherence and Entanglement for Quantum Technology'' project, carried out within the First Team programme of the Foundation for Polish Science co-financed by the European Union under the European Regional Development Fund, and the National Science Centre, Poland, within the QuantERA II Programme (No 2021/03/Y/ST2/00178, acronym ExTRaQT) that has received funding from the European Union's Horizon 2020 research and innovation programme under Grant Agreement No 101017733. CD acknowledges support from the German Federal
Ministry of Education and Research (BMBF) within the funding program ``quantum technologies -- from basic research to market'' in the joint project QSolid (grant
number 13N16163).

\appendix

\section{Entanglement entropy and logarithmic negativity show different ordering}\label{sec:different_ordering}

We will now prove that there exist two states $\ket{\psi}$ and $\ket{\phi}$ such that
\begin{subequations} \label{eq:UnboundedDimension4}
\begin{align}
E(\ket{\psi}) & \geq E(\ket{\phi}), \\
E_{N}(\ket{\psi}) & <E_{N}(\ket{\phi}).
\end{align}
\end{subequations}
For this, we will focus on two-qutrit states of the form
\begin{equation}
    \ket{\Psi}=\sin\alpha\cos\beta\ket{00}+\cos\alpha\cos\beta\ket{11}+\sin\beta\ket{22}.
\end{equation}
For the states $\ket{\psi}$ and $\ket{\phi}$ we choose the parameters $\alpha=1.3$; $\beta=0.75$ and $\alpha=0.7$; $\beta=1$, respectively. With this we find
\begin{subequations}
\begin{align}
E(\ket{\psi}) & =1.195,\,\,\,E(\ket{\phi})=1.157,\\
E_{N}(\ket{\psi}) & =1.324,\,\,\,E_{N}(\ket{\phi})=1.361.
\end{align}
\end{subequations}

Note that there do not exist two-qubit pure states fulfilling Eqs.~(\ref{eq:UnboundedDimension4}), as for two-qubit pure states all entanglement measures have the same ordering.

\section{Continuity of logarithmic negativity} \label{sec:ContinuityNegativity}

In the following, we show that the difference of logarithmic negativities of two bipartite states $\rho$ and $\sigma$ are arbitrarily small if they are arbitrarily close in trace distance.
It is known that trace norm and Hilbert-Schmidt norm satisfy
\begin{equation}\label{norm_relation}
   \Vert M \Vert_{2} \leq \Vert M \Vert_{1} \leq \sqrt{r}\Vert M \Vert_{2},
\end{equation}
for all matrices $M$ and $r$ represents the rank of the matrix $M$. Additionally, we know that Hilbert-Schmidt norm of a matrix is $\Vert M \Vert_2=\left(\sum_{i,j}|M_{ij}|^2\right)^{1/2}$, where $M_{ij}$ are the elements of the matrix $M$. Hence, under partial transposition, the Hilbert-Schmidt norm remains invariant, precisely $\Vert M \Vert_2=\Vert M^{T_A}\Vert_2$. 
Using the relation in Eq. (\ref{norm_relation}), we find 
\begin{eqnarray}
\Vert \rho-\sigma\Vert_2 \leq \Vert \rho-\sigma\Vert_1.
\end{eqnarray}
As Hilbert-Schmidt norm is invariant under partial transpose, we have 
\begin{equation}
    \Vert \rho^{T_A}-\sigma^{T_A}\Vert_2=\Vert \rho-\sigma\Vert_2\leq \Vert \rho-\sigma\Vert_1. 
\end{equation}
Again using Eq. (\ref{norm_relation}), we get 
\begin{align}
    |\Vert \rho^{T_A}\Vert_1-\Vert \sigma^{T_A}\Vert_1|&\leq\Vert \rho^{T_A}-\sigma^{T_A}\Vert_1 \nonumber\\ &\leq \sqrt{r}\Vert \rho^{T_A}-\sigma^{T_A}\Vert_2\nonumber\\ &\leq \sqrt{r}\Vert \rho-\sigma\Vert_1.
\end{align}
Next, we observe that
\begin{align}
E_N(\rho)-E_N(\sigma)&=\log_2\Vert \rho^{T_A} \Vert_1 -\log_2\Vert \sigma^{T_A} \Vert_1 \nonumber\\
&=\log_2\frac{\Vert \rho^{T_A} \Vert_1}{\Vert \sigma^{T_A} \Vert_1}\nonumber\\
&\leq  \log_2 \left(1+\frac{\sqrt{r}\Vert \rho-\sigma\Vert_1}{\Vert \sigma^{T_A} \Vert_1}\right)\nonumber\\
&\leq \frac{\sqrt{r}\Vert \rho-\sigma\Vert_1}{\ln 2 \Vert \sigma^{T_A} \Vert_1}\nonumber\\
&\leq  \frac{\sqrt{d}\Vert \rho-\sigma\Vert_1}{\ln 2 \Vert \sigma^{T_A} \Vert_1}\leq\frac{\sqrt{d}}{\ln 2 }\Vert \rho-\sigma\Vert_1,
\end{align}
where we assume that $\Vert \rho^{T_A} \Vert_1\geq\Vert \sigma^{T_A} \Vert_1$ and consider the relation $\ln(1+x)\leq x$ for $x>-1$ in the second line. In the last line, we further consider $r\leq d$ and $\Vert \sigma^{T_A} \Vert_1\geq 1$, where $d$ is the dimension of the total bipartite state. One can consider the opposite when $\Vert \rho^{T_A} \Vert_1<\Vert \sigma^{T_A} \Vert_1$ and will get 
\begin{equation}
    E_N(\sigma)-E_N(\rho)\leq \frac{\sqrt{d}}{\ln 2 }\Vert \rho-\sigma\Vert_1.
\end{equation}
Now if the states $\rho$ and $\sigma$ are arbitrarily close in trace distance, i.e., $\Vert \rho-\sigma\Vert_1\leq\varepsilon$, then we have 
\begin{equation}
    \left|E_{N}(\rho)-E_{N}(\sigma)\right|\leq\frac{\sqrt{d}\varepsilon}{\ln2}.
\end{equation}
Therefore, we prove that if the trace distance between two states is arbitrarily small, the logarithmic negativities of them are also arbitrarily close to each other.

\section{Proof of theorem \ref{thm:AsymptoticCatalysis}} \label{appendix:asymptotic_catalysis}

Here, we consider a general situation where $\rho$ can be transformed into $\sigma$ via asymptotic catalysis with unit rate. This means, for any $\varepsilon>0$ and any $\delta > 0$ there exist integers $n$ and $m$ with $n>m$, a catalyst state $\tau^C$ and an LOCC protocol $\Lambda$ such that
\begin{subequations} \label{eq:LambdaCatalyticAsymptotic}
\begin{align}
\left\Vert \Lambda\left[\rho^{\otimes n}\otimes\tau^{C}\right]-\sigma^{\otimes m}\otimes\sigma_g\otimes\tau^{C}\right\Vert _{1} & \leq\varepsilon,\label{asymptotic}\\
\mbox{Tr}_{S^{\otimes n}}\left[\Lambda\left(\rho^{\otimes n}\otimes\tau^{C}\right)\right] & =\tau^{C},\\
\frac{m}{n}+\delta &\geq 1. \label{rate}
\end{align}
\end{subequations}
To simplify the notation, we introduce the state $\mu = \Lambda\left(\rho^{\otimes n}\otimes\tau^{C}\right)$ which is in $S^{\otimes n}\otimes C$, where $S$ and $C$ denote the Hilbert space of the system and the catalyst, respectively. Here, $\sigma_g$ is a product state acting on the Hilbert space $S^{\otimes(n-m)}$.

Now, we will show that a catalyst $\tau'$ and an LOCC operation $\Lambda'$ can be chosen, such that 
$\Vert \mu' - \sigma\otimes \tau'\Vert_{1}< 2(\varepsilon + \delta)$,
where
\begin{equation}
\mu'= \Lambda'(\rho\otimes\tau')
\end{equation}
and
\begin{equation}
    \Tr_S[\mu'] = \tau'.
\end{equation}
We set the catalyst state to be 
\begin{equation}
    \tau' = \frac{1}{n}\sum_{k=1}^{n}\rho^{\otimes k-1} \otimes \mu_{n-k} \otimes \ket{k}\!\bra{k}, \label{eq:tau}
\end{equation}
where $\mu_{i}$ is the reduced state of $\mu$ after tracing out systems $i+1$ to $n$. Note that $\mu_i$ is a state in the Hilbert space of $S_1$ to $S_i$ and $C$. Additionally, note that $\mu_n\equiv\mu$, $\mu_0\equiv\tau^C$, and $\rho^{\otimes 0}\equiv 1$. The Hilbert space of the catalyst state $\tau'$ is in $S^{ \otimes{n-1}}\otimes C\otimes K$, where $K$ represents the Hilbert space of an auxiliary system of dimension $n$ which is maintained by Alice. 

In the following, we give a 3-step construction of an LOCC protocol $\Lambda'$, similar to the one in \cite{Kondra2102.11136,shiraishi2020quantum}:

(i) In the first step, Alice measures her register $K$ in the basis $\ket{k}$ and communicates the outcome to Bob. If the outcome is $n$, Alice and Bob perform the LOCC protocol $\Lambda$ given in Eqs.~(\ref{eq:LambdaCatalyticAsymptotic}) on $S_{1}\otimes S_{2}\otimes\cdots\otimes S_{n}\otimes C$. If the outcome is different from $n$, the parties do nothing.

(ii) Alice applies a unitary on her register $K$, transforming $\ket{n}\xrightarrow{}\ket{1}$ and $\ket{i}\xrightarrow{}\ket{i+1}$.

(iii) Finally, both Alice and Bob, apply a  SWAP on their parts of ($S_{i}$, $S_{i+1}$) and  ($S_{n}$, $S_{1}$), shifting $S_{i}\xrightarrow{} S_{i+1}$ and $S_{n}\xrightarrow{} S_{1}$.

The initial state of the system along with the catalyst is 
\begin{equation}
    \rho \otimes \tau' = \frac{1}{n}\sum_{k=1}^{n}\rho^{\otimes k} \otimes \mu_{n-k} \otimes \ket{k}\bra{k}.
\end{equation}
The initial state after applying step (i) becomes
\begin{equation}
    \eta^{i} = \frac{1}{n}\sum_{k=1}^{n-1}\rho^{\otimes k} \otimes \mu_{n-k} \otimes \ket{k}\bra{k} + \frac{1}{n}\mu\otimes \ket{n}\bra{n}.
\end{equation}
After step (ii), $\eta^{i}$ is transformed into 
\begin{equation}
    \eta^{ii} = \frac{1}{n}\sum_{k=1}^{n}\rho^{\otimes k-1} \otimes \mu_{n+1-k} \otimes \ket{k}\bra{k}. \label{eq:Muii}
\end{equation}
Having traced out $S_{n}$ from $\eta^{ii}$, we obtain $\tau'$, which is the initial state of the catalyst, see Eq.~(\ref{eq:tau}). Hence, we perform step (iii) to transform $\eta^{ii}$ to the final state $\mu'$ having the property $\Tr_{S}[\mu']$ = $\tau'$. Therefore, the state of the catalyst remains unchanged in the above procedure. 

To complete the proof, we will now show that $\Vert \mu' - \sigma\otimes \tau'\Vert_{1}< 2(\varepsilon +\delta)$.  As $\mu'$ is equivalent to the state $\eta^{ii}$ up to a cyclic SWAP, we have 
\begin{equation}
\Vert \mu' - \sigma\otimes \tau'\Vert_{1}=  \Vert \eta^{ii} - \gamma\Vert_{1},   
\end{equation}
where 
\begin{equation}
    \gamma=\frac{1}{n}\sum_{k=1}^{n}\rho^{\otimes k-1} \otimes \tilde{\mu}_{n+1-k} \otimes \ket{k}\!\bra{k}.
\end{equation}
Here, $\tilde{\mu}_{i}$ is constructed from $\mu_i$ by tracing out $S_i$ and replacing it with $\sigma$, i.e., 
$\tilde{\mu}_i = (\Tr_{S_i} \mu_i) \otimes \sigma$,
up to the order of the components in the tensor product:
$S_{1}\otimes S_{2}\otimes\cdots\otimes S_{i}\otimes C$.
Now, we obtain
\begin{align}
\Vert\eta^{ii}-\gamma\Vert_{1} & =\frac{1}{n}\sum_{k=1}^{n}\Vert\mu_{n+1-k}-\tilde{\mu}_{n+1-k}\Vert_{1}\\
 & =\frac{1}{n}\sum_{k=1}^{n-m}\Vert\mu_{n+1-k}-\tilde{\mu}_{n+1-k}\Vert_{1}\nonumber \\
 & +\frac{1}{n}\sum_{k=n-m+1}^{n}\Vert\mu_{n+1-k}-\tilde{\mu}_{n+1-k}\Vert_{1}\nonumber \\
 & \leq2\delta+\frac{1}{n}\sum_{l=1}^{m}\Vert\mu_{l}-\tilde{\mu}_{l}\Vert_{1}\nonumber \\
 & \leq2\delta+\frac{1}{n}\sum_{l=1}^{m}\Vert\mu_{l}-\sigma^{\otimes l}\otimes\tau^{C}\Vert_{1}\nonumber \\
 & +\frac{1}{n}\sum_{l=1}^{m}\Vert\tilde{\mu}_{l}-\sigma^{\otimes l}\otimes\tau^{C}\Vert_{1}\nonumber \\
 & \leq2(\delta + \frac{m}{n} \varepsilon) \leq2(\delta + \varepsilon).\nonumber 
\end{align}
In the first inequality we used Eq.~(\ref{rate}) and the fact that $||\rho-\sigma||_1 \leq 2$ for any quantum states $\rho$ and $\sigma$. The second inequality follows from the triangle inequality. In the third inequality we used Eq. (\ref{asymptotic}) together with the fact that the trace norm does not increase under partial trace. 

The above arguments prove that it is possible to convert $\rho$ into $\sigma$ via catalytic LOCC whenever the conversion $\rho \rightarrow \sigma$ is possible via asymptotic catalysis with unit rate. To prove the converse, note that whenever there exists a catalytic LOCC protocol converting $\rho$ into $\sigma$, it is clearly possible to achieve asymptotic catalytic conversion with unit rate. This completes the proof of the theorem. 

We note that this theorem also covers the result presented in Theorem 1 of~\cite{Kondra2102.11136}. To see this, note that if $\rho$ can be converted into $\sigma$ via asymptotic LOCC with unit rate, it is also possible to achieve conversion with unit rate via asymptotic catalysis, simply by adding a catalyst which does not take part in the process. Similarly as in Theorem 1 of~\cite{Kondra2102.11136}, the proof presented above also applies for the case where the system $S$ consists of more than two subsystems, and multipartite LOCC protocols are considered. 

\section{Maximum entropy for a probability distribution where one of the probabilities is given}
\label{sec:maxEntropy}
We will now prove that the entropy $H(p_i)$ of a probability distribution $\{p_i\}_{i=0}^3$ achieves its maximum for a given $p_0$ when $p_1=p_2=p_3=(1-p_0)/3$. 

For $j=\{1,2,3\}$ we define the probability distribution $q_j=p_j/(1-p_0)$. Using this we find
\begin{equation}
-\sum_{j=1}^{3}q_{j}\log_{2}q_{j}=\log_{2}(1-p_0)-\frac{\sum_{j=1}^{3}p_{j}\log_{2}p_{j}}{1-p_0}
\end{equation}
which implies
\begin{align}
-\sum_{j=1}^{3}p_{j}\log_{2}p_{j} = & -(1-p_0)\sum_{j=1}^{3}q_{j}\log_{2}q_{j} \\
 & -(1-p_0)\log_{2}(1-p_0). \nonumber
\end{align}
For a given $p_0$ the maximum of the right hand side is achieved when the distribution is uniform, precisely $q_j=1/3$, and hence, $p_j=(1-p_0)/3$ maximizes the entropy $H(p_i)$, as claimed.

\section{Comparison with the fidelity of catalytic teleportation \cite{Lipka-Bartosik2102.11846}}\label{appendix_comparison}

From Theorem 1 of Ref. \cite{Lipka-Bartosik2102.11846}, one can see that, a bipartite state $\rho$  (shared between Alice and Bob) can be used to teleport $m$-qubits perfectly whenever
\begin{equation}\label{cond1}
   f_{\text{reg}}(\rho)=\lim_{n\to\infty}\frac{f_{n}(\rho^{\otimes n})}{n}\geq 1.
\end{equation}
Where, $f_{n}(\sigma)$ is given by the following optimisation problem
\begin{equation}
     f_{n}(\sigma)=\max_{\mathcal{E}\in \text{LOCC}}\sum^{n}_{i=1}\bra{\phi_{2^m}^{+}}\text{Tr}_{/i}\left(\mathcal{E}(\sigma)\right)\ket{\phi_{2^m}^{+}}.
\end{equation}
Here, $\text{Tr}_{/i}$ is the partial trace performed over systems $1.....i-1,i+1....n$. Note that Eq. (\ref{cond1}) is equivalent to $R_{\textrm{mar}}(\rho\to\ket{\phi^{+}_2})\geq m$ i.e, \emph{marginal rate} of transforming $\rho$ into $\ket{\phi^{+}_2}$ is greater than $m$. We refer to Ref. \cite{ganardi2023catalytic} for the definition of marginal rate. From Proposition 6 of \cite{ganardi2023catalytic} we see that
\begin{equation}
    R_{\textrm{mar}}(\rho\to\ket{\phi^{+}_2})=R(\rho\to\ket{\phi^{+}_2}),
\end{equation}
where $R(\rho\to\ket{\phi^{+}_2})$ is the asymptotic transformation rate from $\rho$ into $\ket{\phi^{+}_2}$. This shows Eq. (\ref{cond1}) is equivalent to
\begin{equation}
    E_{\mathrm{d}}(\rho)\geq m.
\end{equation}
This shows the connection between Eq. (\ref{eq:EdM}) and the main result of reference \cite{Lipka-Bartosik2102.11846}.


\bibliographystyle{quantum}
\bibliography{literature}

\end{document}